\documentclass[lettersize,journal,balance,twocolumn]{IEEEtran}
\usepackage{amsmath,amsfonts}
\usepackage{algorithmic}
\usepackage{algorithm}
\usepackage{array}
\usepackage[caption=false,font=normalsize,labelfont=sf,textfont=sf]{subfig}
\usepackage{textcomp}
\usepackage{stfloats}
\usepackage{url}
\usepackage{verbatim}
\usepackage{graphicx}
\usepackage{cite}

\usepackage[hidelinks]{hyperref}
\usepackage{mathtools}
\usepackage{nicematrix}
\usepackage{balance}
\usepackage{enumitem}
\usepackage{bm}
\usepackage{amsthm}
\usepackage{amssymb}
\usepackage{multirow}
\usepackage{hhline}
\usepackage{physics}
\usepackage{pifont}

\newtheorem{lemma}{Lemma}

\newtheorem{remark}{Remark}

\hypersetup{ %
    breaklinks=true, %
    citecolor=black, %
    colorlinks=true, %
    linkcolor=black, %
    urlcolor=blue
 }

\newcolumntype{P}[1]{>{\centering\arraybackslash}p{#1}}
\newcolumntype{M}[1]{>{\centering\arraybackslash}m{#1}}
\newcommand{\cmark}{\ding{51}}%
\begin{document}

\title{Key Reconciliation with RC-LDPC/Error Estimation for Satellite-based FSO/QKD Systems}

\author{
Cuong T. Nguyen,~\IEEEmembership{Graduate Student Member,~IEEE,} 
Hoang D. Le,~\IEEEmembership{Member,~IEEE,} 
Anshul Jaiswal,~\IEEEmembership{Senior Member,~IEEE,}
Swaminathan R.,~\IEEEmembership{Senior Member,~IEEE,}
and Anh T. Pham,~\IEEEmembership{Senior Member,~IEEE}


\thanks{\textcopyright 2026 IEEE. Personal use of this material is permitted. Permission from IEEE must be obtained for all other uses, in any current or future media, including reprinting/republishing this material for advertising or promotional purposes, creating new collective works, for resale or redistribution to servers or lists, or reuse of any copyrighted component of this work in other works. DOI: \href{https://doi.org/10.1109/TVT.2026.3725740}{10.1109/TVT.2026.3725740}}

\thanks{This work was supported by the Telecommunications Advancement Foundation (TAF) and the Japan Society for the Promotion of Science (JSPS) KAKENHI under Grant Number 24K14918.}

\thanks{Cuong T. Nguyen, Hoang D. Le, and Anh T. Pham are with the Department of Computer Science and Engineering, The University of Aizu, Fukushima 965-8580, Japan (e-mail: cuong98hp@gmail.com; hoangle@u-aizu.ac.jp; pham@u-aizu.ac.jp).}

\thanks{Anshul Jaiswal is with the Department of Electronics and Communication Engineering, IIT Roorkee, Roorkee 247667, India (e-mail: anshul.jaiswal@ece.iitr.ac.in).}

\thanks{Swaminathan R. is with the Department of Electrical Engineering, Indian Institute of Technology Indore, Indore, India (e-mail: swamiramabadran@iiti.ac.in).}

\thanks{The corresponding author: Hoang D. Le}

}

\markboth{IEEE Journal,~Vol.~XX, No.~Y, Aug.~2026}%
{Shell \MakeLowercase{\textit{et al.}}: A Sample Article Using IEEEtran.cls for IEEE Journals}


\maketitle

\begin{abstract}
Satellite-based free-space optics (FSO) quantum key distribution (QKD) systems have recently attracted significant research interest due to their potential to enable globally secured applications. However, the inherent uncertainty of FSO channels, caused by weather conditions and satellite mobility, induces severe fluctuations in quantum bit-error rate (QBER) between legitimate users. This makes designing an efficient key reconciliation, an essential step in the QKD post-processing, particularly challenging. In this work, we propose a key reconciliation scheme that combines protograph rate-compatible (RC) low-density parity-check (LDPC) codes with a syndrome-based error estimation method. The proposed error estimation method reduces the number of communication rounds without requiring additional information disclosure. Furthermore, to our best knowledge, an analytical framework is first developed to evaluate end-to-end secret-key throughput (SKT), accounting for the impact of imperfect error estimation and dynamic FSO channel conditions. Numerical results demonstrate that the proposed scheme consistently outperforms conventional blind reconciliation under diverse FSO channel conditions and provide practical guidelines for system parameter selection. Finally, we validate the proposed framework through a case study that incorporates data from a Starlink low-Earth orbit (LEO) satellite and moving ground vehicles.
\end{abstract}


\section{Introduction}
\label{sec:introduction}

In recent years, satellite-based quantum key distribution (QKD) has attracted growing global interest, as evidenced by missions such as SpeQtral-1 \cite{speqtral}, Q4S mission \cite{q4s}, and QEYSSat \cite{jennewein2023qeyssat}. QKD enables the sharing of information-theoretically secret keys by exploiting quantum principles, making it resilient to advances in computational technologies \cite{scarani2009security}. When integrated with satellite-based free-space optical (FSO) links, QKD can support global, secure, and wireless applications, including the Internet of Vehicles (IoVs). The feasibility of such systems has been demonstrated through landmark experiments, notably Micius \cite{lu2022micius} and SOTA \cite{takenaka2017satellite}. Consequently, satellite-based FSO/QKD is regarded as a cornerstone of the future quantum-secured Internet \cite{trinh2024quantum}.

A QKD system operates in two phases: quantum transmission and post-processing. During the quantum phase, two legitimate users (Alice and Bob) exchange quantum states to generate correlated key materials. In post-processing,  sifting produces correlated but imperfect keys due to quantum channel noise and potential eavesdropping, characterized by the quantum bit-error rate (QBER). These mismatches are corrected during key reconciliation via authenticated public communication \cite{brassard1993secret}. Syndrome-based reconciliation using low-density parity-check (LDPC) codes is widely adopted and forms the focus of this work. As key reconciliation is the most computationally intensive step and often the throughput bottleneck in QKD systems, efficient reconciliation design is critical for satellite-based FSO/QKD deployments.

This paper proposes a blind reconciliation scheme for satellite-based FSO/QKD systems and develops a comprehensive end-to-end secret-key throughput (SKT) framework. The main contributions are summarized as follows:

\begin{itemize}
    \item First, we propose a joint design of blind key reconciliation using protograph rate-compatible LDPC (RC-LDPC) codes combined with a syndrome-based error estimation scheme tailored to satellite-based FSO/QKD systems. While these components have been individually studied, their joint optimization for addressing satellite-specific challenges constitutes the key contribution.
    \item Second, we develop, to the best of our knowledge, the first comprehensive theoretical framework for end-to-end SKT analysis, covering both conventional blind reconciliation and the proposed scheme. The framework incorporates major FSO impairments, including atmospheric turbulence, cloud coverage, and pointing errors, as well as imperfections arising from syndrome-based error estimation. Monte Carlo simulations validate the analytical results.
    \item Furthermore, we characterize the statistical behavior of block-wise QBER by deriving its probability density function (PDF) via curve fitting and evaluating candidate distributions using goodness-of-fit metrics. The resulting PDF is used for system-level performance analysis.
    \item Finally, numerical results provide practical insights through a case study based on Starlink low-Earth-orbit (LEO) satellites and ground-vehicle data, demonstrating the effectiveness of the proposed design.
\end{itemize}

The remainder of this paper is organized as follows. Section~\ref{sec:related_work} reviews the related work and discusses our motivation for this study. Section~\ref{sec:system_desp} describes the overall system, the FSO channel model, and the considered QKD protocol. The details of the proposed design are presented in Section~\ref{sec:kr_design}. Section~\ref{sec:performance_analysis} provides the theoretical analysis of the proposed design in terms of SKT. Numerical results and discussions are given in Section~\ref{sec:results_discuss}. Finally, Section~\ref{sec:conclusion} concludes the paper.

\section{Related Work}
\label{sec:related_work}
\begin{table*}[ht!]
\centering
\caption{Contributions Comparison of Related Work (\cmark: Considered).}
\label{tab:related_works}
\begin{tabular}{| M{1.3cm} | M{1.8cm} | M{1.8cm} | M{1.8cm} | M{1.8cm} | M{2cm} | M{1.8cm} | M{1.8cm} |}
\hline
\multirow{3}{*}{Works} & \multicolumn{2}{ c |}{Considered Systems} & \multicolumn{3}{ c |}{Key Reconciliation Design} & \multicolumn{2}{ c |}{Theoretical Framework} \\ \cline{2-8}
  & Satellite-based QKD System & Practice (e.g., Starlink’s data) & Blind Reconciliation & Protograph LDPC Code & Syndrome-based Error Estimation & Throughput Analysis  & Imperfect Error Estimation  \\ 
\hline
\cite{ai2017short, ai2018quantum, ai2020reconciliation} & \cmark &  & & &  &  &   \\ 
\hline
\cite{martinez2012blind, kiktenko2017symmetric, liu2020blind, kiktenko2021blind, mao2021high, borisov2022asymmetric} &  & & \cmark & &  &  &  \\ 
\hline
\cite{tarable2024rateless} &  & & \cmark & \cmark &  &  &  \\ 
\hline
\cite{kiktenko2018error, gao2019multi} & & &  &  & \cmark &  &   \\ 
\hline
\cite{cuong2024blind} & \cmark & \cmark & \cmark & \cmark & &  &  \\ 
\hhline{|=|=|=|=|=|=|=|=|}
\textbf{This study} & \cmark & \cmark & \cmark & \cmark & \cmark & \cmark & \cmark \\ 
\hline
\end{tabular}
\end{table*}

Table~\ref{tab:related_works} presents related work discussed in this Section in terms of key reconciliation design and theoretical frameworks on performance evaluation.

First, a prime concern in the reconciliation design is the selection of appropriate code rates to correct all errors while minimizing the number of exchanged syndrome bits. This issue becomes especially critical in the satellite-based FSO/QKD systems due to the inherent high latency and the highly fluctuating QBER \cite{capraro2012impact}. To tackle this issue, existing works have focused on developing adaptive-rate reconciliation schemes for satellite-based FSO/QKD systems \cite{ai2017short, ai2018quantum, ai2020reconciliation}. In this design, both legitimate users share a set of LDPC codes with different code rates. A proper code rate will be chosen to correct errors in one communication round based on an accurately estimated QBER. Notably, Xiaoyu \textit{et al.} considered the rate-adaptive reconciliation with short-length LDPC codes for satellite-implemented device-independent QKD systems \cite{ai2017short}. A comparison between fixed-rate and adaptive-rate reconciliation for satellite-based QKD systems was presented in \cite{ai2018quantum}. As expected, the post-decoding QBER of the rate-adaptive approach outperforms that of the fixed-rate approach over a wide range of QBER values. Most recently, the authors in \cite{ai2020reconciliation} investigated a reconciliation strategy to maximize the decoding rate under various constraints for real-time satellite-based QKD systems. 

Existing works \cite{ai2017short, ai2018quantum, ai2020reconciliation} relied on the assumption of an accurate QBER estimator, which is essential for the effective design of adaptive-rate reconciliation schemes. In practice, error estimation is imperfect, and the estimated QBER can be underestimated or overestimated. Both cases can result in unnecessary excess information leakage, especially in the case of underestimation, where a new code rate is adopted. This can sometimes result in the discarding of sifted blocks. To mitigate the effect of imperfect error estimation, an alternative is blind reconciliation \cite{martinez2012blind}. The key difference between this method and the adaptive-rate method is the utilization of the rate-compatible (RC)-LDPC code family. This structure enables the deployment of lower-rate codes by reusing information from higher-rate ones. As a result, if underestimation happens, blind reconciliation can gradually lower the code rate until it reaches the proper level. The initial design of blind reconciliation proposed in \cite{martinez2012blind} considers the shortening and puncturing techniques to construct the RC-LDPC code family. Specifically, Alice first creates an extended key by adding random bits to her sifted key. These bits are treated as punctured and shortened bits during Bob's syndrome decoding. If the decoding fails, Alice reveals additional punctured bits to Bob via the public channel, thus lowering the code rate. This gives Bob a higher chance of successful decoding.

Following the design in \cite{martinez2012blind}, extensive studies on blind reconciliation have been widely investigated for QKD-based optical fiber systems \cite{kiktenko2017symmetric, liu2020blind, kiktenko2021blind, mao2021high, borisov2022asymmetric, tarable2024rateless}. In \cite{kiktenko2017symmetric}, a symmetric blind reconciliation design was proposed to improve the system efficiency compared to the conventional approach. Zhihong \textit{et al.} considered the step-size varying technique in the blind reconciliation design, thus reducing the operation time \cite{liu2020blind}. The authors in \cite{kiktenko2021blind} proposed a blind reconciliation design based on polar codes. Mao \textit{et al.} proposed a method to minimize information leakage by reusing syndrome bits \cite{mao2021high}. In \cite{borisov2022asymmetric}, a key reconciliation design based on blind reconciliation and prior error estimation for industrial QKD systems was proposed. In \cite{tarable2024rateless}, the authors proposed a design and construction method for protograph LDPC codes, specifically for blind reconciliation. Therein, the protograph LDPC code refers to a structured LDPC code class constructed from a small graph with a few elements. Compared to unstructured LDPC codes, this type of LDPC code can offer efficient hardware implementation and flexible rate adaptivity when the QBER fluctuates significantly \cite{cuong2024blind}.

Aside from our report in \cite{cuong2024blind}, the design of the protograph RC-LDPC-based blind reconciliation for satellite-based FSO/QKD has not been investigated in the literature. The design in \cite{cuong2024blind}, nevertheless, does not account for error estimation. Specifically, the reconciliation process will always begin with the highest code rate and gradually decrease the rate until successful decoding is achieved or all code members have been used. In some situations, the design may require multiple communication rounds to achieve the optimal code rate. This becomes especially critical in real-time scenarios where post-processing is conducted directly over the high-latency satellite-to-ground public channel. To tackle this issue, it is necessary to incorporate error estimation into the design.

A common approach is random sampling, in which both sides reveal a fraction of their sifted blocks via the public channel and compare them to estimate QBER. As these bits are known to Eve, they will be discarded afterward \cite{wolf2021quantum}. Thus, this method reduces the key generation rate and may not be suitable when the rate is already low due to high loss over the quantum channel. Moreover, it always requires a round of exchanges to reveal bits, which may be inefficient when the average QBER is low. Another type of scheme utilizes QBER values of previous sifted blocks to estimate the QBER of the current block \cite{borisov2022asymmetric, kiktenko2016post}. However, these methods can be inaccurate for satellite-based QKD links, where fading is more severe than in fiber-based QKD links and the QBER per sifted block fluctuates significantly. Recently, syndrome-based error estimation methods have been proposed and applied to QKD systems \cite{kiktenko2018error, gao2019multi}. The key idea is that the QBER of each sifted block can be estimated based on the exchanged syndrome bits. As shown, this method, when used in conjunction with blind reconciliation, can operate properly over FSO-based QKD links without leaking additional information. As a result, this method presents a potential solution for designing blind reconciliation with protograph RC-LDPC codes in satellite-based FSO/QKD systems.

To comprehensively evaluate the proposal's effectiveness, it is essential to assess end-to-end secret-key throughput (SKT). This metric reflects the average number of secret bits processed by the QKD system per second, considering both the quantum and post-processing phases. Moreover, it is crucial to provide a comprehensive theoretical framework to gain insight into the end-to-end performance of satellite-based FSO/QKD systems. Conducting this task is daunting and not straightforward due to the numerous challenges involved. Additionally, error estimation is not always accurate, necessitating a thorough investigation into the impact of this factor on system performance. To the best of our knowledge, a design and theoretical framework of this nature has not been previously available in the literature, which motivates us to conduct this work.

\begin{figure}[t]
    \centering
    \includegraphics[width=.9\columnwidth]{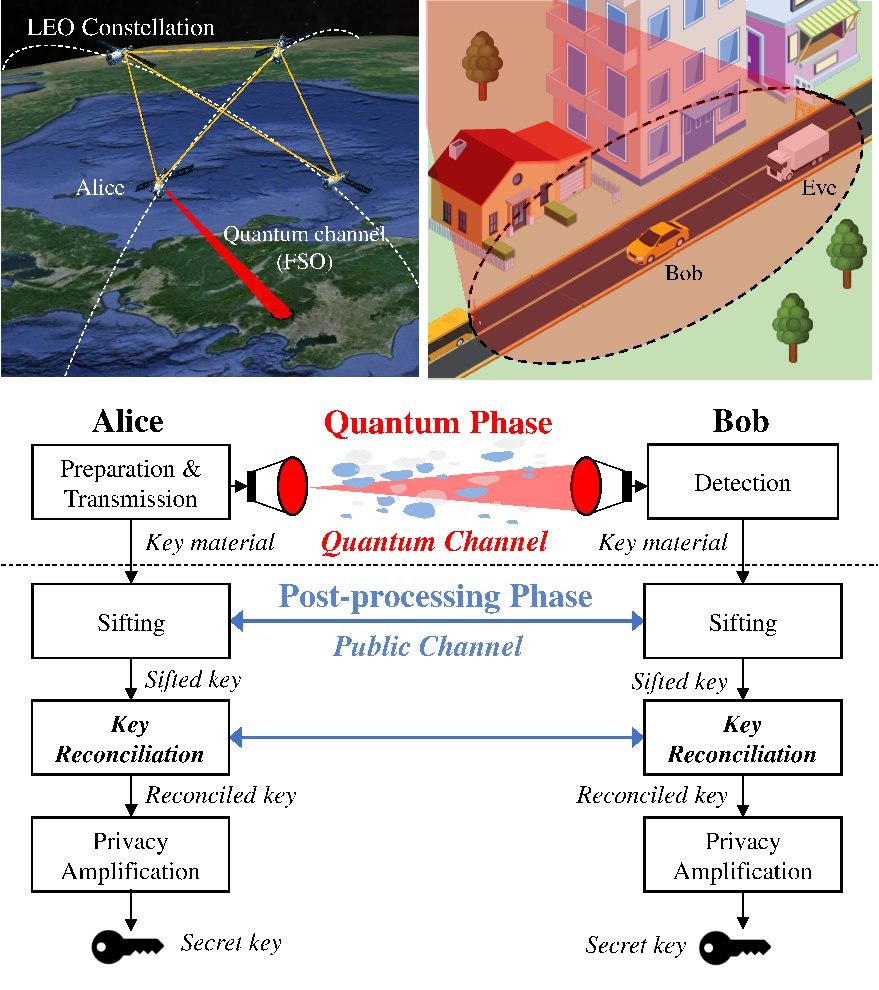}
    \caption{LEO satellite-based FSO/QKD system for secure IoV applications and the block diagram of QKD protocols.}
    \label{fig:system_model}
\end{figure}

\section{System Description}
\label{sec:system_desp}

In this section, we first describe the satellite-based FSO/QKD system supporting secure IoV applications. We then present the LEO-to-vehicle FSO channel model. Finally, we review the considered non-coherent implementation of the CV-QKD protocol.

\subsection{System Model}
\label{sec:system_overview}

Figure~\ref{fig:system_model} presents the LEO satellite-based FSO/QKD system for secure IoV applications. In particular, we focus on the scenario of secure LEO satellite-to-vehicle communications in the context of Internet from space. Here, an LEO satellite (Alice) attempts to distribute secret keys to a ground vehicle (Bob). We consider unauthorized receiver attacks (URA), in which a ground eavesdropper (Eve) attempts to tap downlink quantum signals by being within the beam footprint. Eve can also listen to information exchanged via the public channel.
\begin{figure*}[t]
    \centering
    \includegraphics[width=14 cm]{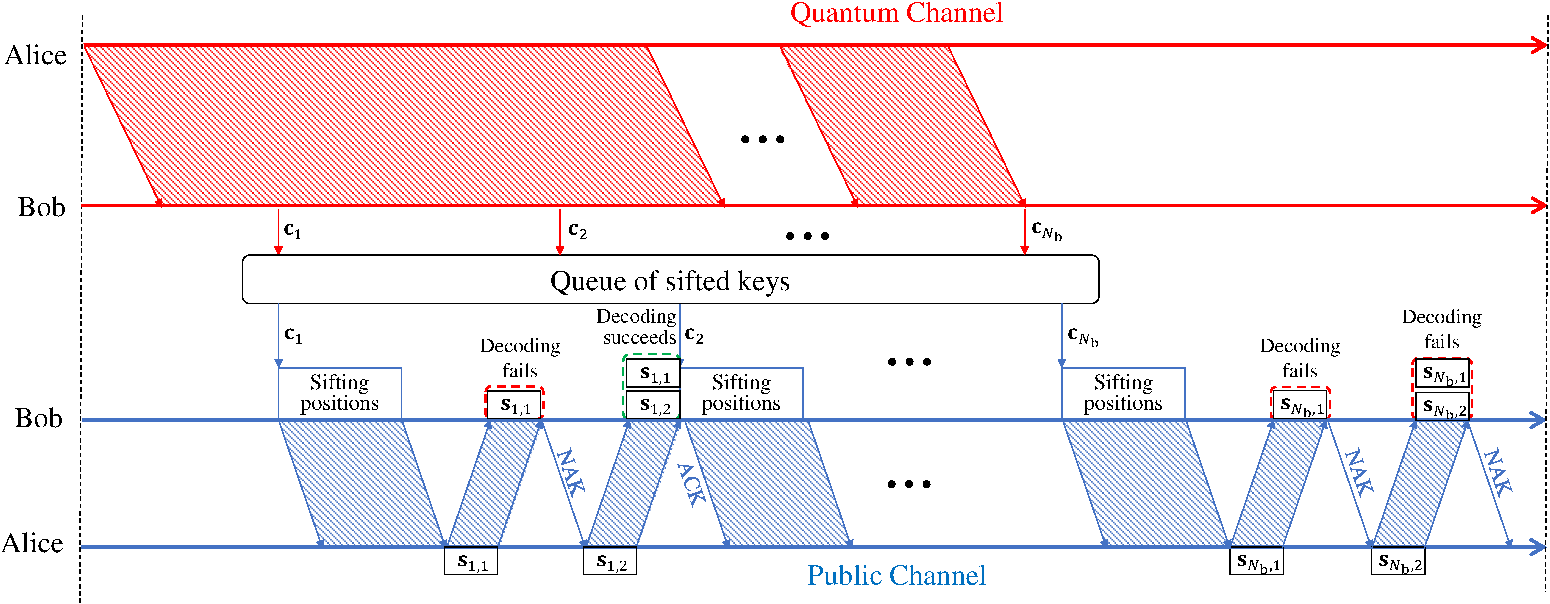}
    \caption{An example of a superframe of the considered system, where the quantum phase is conducted via the quantum channel and the post-processing phase is conducted via the public channel.}
    \label{fig:superframe}
\end{figure*}

The considered system is divided into two phases: (i) {\it quantum phase} and (ii) {\it post-processing phase}. In addition, we divide time into superframes, defined as the time duration during which both sides obtain and reconcile a predefined number of sifted blocks, denoted as $N_\text{b}$. In each superframe, the quantum phase and the post-processing phase are executed simultaneously. This indicates that a superframe begins at the instant Alice starts sharing the key material via the quantum channel. It continues until the last sifted block finishes the reconciliation step via the public channel. Fig.~\ref{fig:superframe} exemplifies a superframe of the desired system, where the quantum phase and the post-processing phase are conducted over the quantum FSO channel and the public channel, respectively.

\textbf{In the quantum phase,} Alice generates key materials and shares them with Bob via an FSO channel. Here, we employ the non-coherent CV-QKD protocol using SIM/BPSK with a dual-threshold detector \cite{trinh2018design} (refer to Sec.~\ref{sec:cv_qkd_dtdd} for a detailed description). We assume that the coherence time of the channel is $t_\text{coh}$. This implies that the fading coefficients $h \left( t \right)$ change every $t_\text{coh}$ seconds. After gathering $n_\text{sift}$ sifted bits, Bob forms a sifted block $F_i$ $\left(1 \leq i \leq N_\text{b} \right)$ and puts it in a first-in-first-out queue. Once Bob collects $N_\text{b}$ sifted blocks, he stops receiving the key materials until the next superframe.

\textbf{In the post-processing phase,} Alice and Bob transform the shared key material into secure keys via several steps, i.e., sifting, key reconciliation, and privacy amplification. These steps are conducted via an authenticated and public channel. The first two steps are executed sequentially and independently on each sifted block in the queue. \textit{Regarding the sifting step,} for each sifted block, Bob transmits information about timestamps, in which he can detect the signals, to Alice. We assume that each timestamp is a real value composed of 64 bits \cite{almeida2024classical}. As a result, the amount of information needed to be transmitted via the public channel for the sifting step is $n_\text{sp} = 64 n_\text{sift}$. After receiving this information, Alice forms her sifted block accordingly.

\textit{Regarding the key reconciliation step,} we consider the blind reconciliation with protograph LDPC code and error estimation, which is detailed in Sec.~\ref{sec:kr_design}. If the reconciliation fails, both sides discard their sifted blocks. Otherwise, they store their reconciled blocks in the buffer. \textit{As for the privacy amplification step,} the purpose of this step is to remove information potentially leaked to Eve throughout the procedure. A common method is compressing the reconciled key into a shorter yet more secure key using a shared hash function chosen from a universal hash family \cite{li2018memory}. Particularly, after gathering enough reconciled blocks, $N_\text{sb}$, both sides concatenate all of them to construct a long block of length $n_\text{PA} = n_\text{sift} N_\text{sb}$ bits. They then independently perform universal hashing on the constructed block to obtain the secure key \cite{milicevic2017key}. The assumptions used in this paper are as follows: (i) the public channel is assumed to be authenticated and error-free\footnote{Authentication over the public channel can be guaranteed by attaching verification hash digests \cite{kiktenko2016post}, while reliability is ensured via robust error-correction codes \cite{le2023fso}.}, and (ii) the probability that the reconciliation fails and that this failure goes undetected by the parties is negligible\footnote{This probability can be reduced to an arbitrary level by exchanging the hash digests of the reconciled blocks \cite{leverrier2010finite}.}.
\subsection{FSO Channel Model}
\label{sec:channel_model}

As for the LEO-to-vehicle FSO channel model, we consider three major impairments\footnote{ Note that existing optical receiver design can effectively compensate for the Doppler effect between an LEO satellite and a moving ground vehicle \cite{cuong2024toward}.}, i.e., cloud attenuation $h_\text{c}$, atmospheric turbulence $h_\text{t}$, and pointing errors-induced fluctuations $h_{\text{p}, X}$, where $X \in \left\{\text{B, E} \right\}$ is used to distinguish between the coefficients of Bob and Eve, respectively. Then, the composite channel coefficient is given as $h_X = h_\text{c} h_\text{t} h_{\text{p}, X}$. Details of these impairments are as follows.

\begin{figure*}[t]
    \centering
    \includegraphics[width=\textwidth]{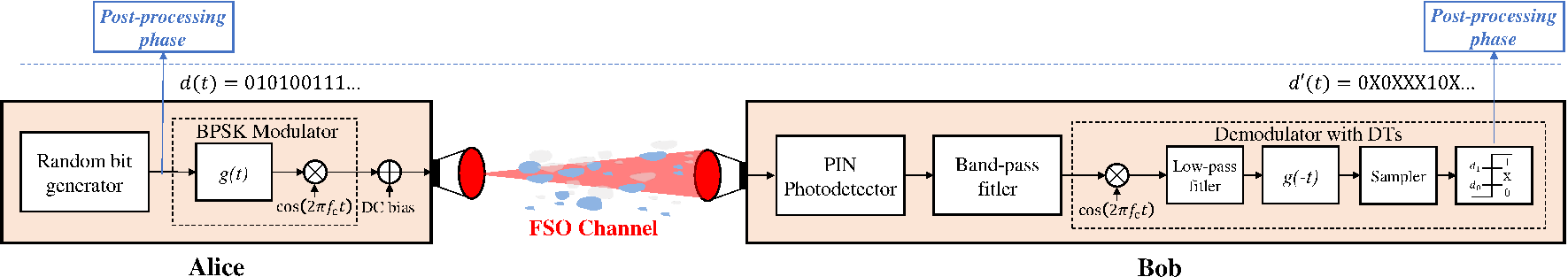}
    \caption{Block diagram of the free-space CV-QKD protocol using SIM/BPSK and a dual-threshold detector implemented on the quantum phase.}
    \label{fig:phy_layer}
\end{figure*}

\subsubsection{Cloud Attenuation}
\label{sec:cloud_attenuation}
Due to the presence of liquid water particles in clouds, optical signals may experience a scattering effect, resulting in a reduction in optical power. According to \cite{ghassemlooy2019optical}, the cloud attenuation coefficient is given as
\begin{align}
    \label{eqn:cloud_attenuation}
    h_\text{c} = \exp \left[ -\sigma H_\text{c} \sec \left( \xi \right) \right],
\end{align}
where $H_\text{c}$ is the vertical extent of clouds, and $\xi$ is the zenith angle. Moreover, $\sigma$ is the attenuation coefficient expressed as a function of the optical wavelength $\lambda$ and the visibility $V$ \cite[(12)]{le2023fso}. In addition, $V = 1.002 / \left( M_\text{c} N_\text{c} \right)^{0.6473}$, where $M_\text{c}$ $\left( \text{g/m}^{-3} \right)$ is the cloud liquid water content (CLWC), and $N_\text{c}$ $\left( \text{cm}^{-3} \right)$ is the cloud droplet number concentration \cite{le2023fso}.

\subsubsection{Atmospheric Turbulence}
\label{sec:turbulence}
This phenomenon causes the scintillation effect, resulting in fluctuations in the received power at the receiver detector. To model a wide range of turbulence strengths, we consider the Gamma-Gamma (GG) model, whose probability density function (PDF) is given as \cite{andrews2005laser}
\begin{align}
    \label{eqn:gg_fading}
    f_{\text{GG}} \left( h_\text{t} \right) = \frac{2 \left( \alpha \beta \right)^{\frac{\alpha + \beta}{2}}}{\Gamma \left( \alpha \right) \Gamma \left( \beta \right)} h_\text{t}^{\frac{\alpha + \beta}{2} - 1} K_{\alpha - \beta} \left(2 \sqrt{\alpha \beta h_\text{t}} \right), 
\end{align}
where $\Gamma \left( . \right)$ is the Gamma function, and $ K_{\alpha - \beta} \left( . \right)$ presents the modified Bessel function of the second kind of order $\alpha - \beta$. Moreover, the scale parameters $\alpha$ and $\beta$ denote the effective number of large-scale and small-scale cells of the scattering process in the atmosphere, respectively. These parameters can be calculated based on the Rytov variance $\sigma_\text{R}^{2}$, as in \cite[(15)]{le2023fso}.

\subsubsection{Generalized Pointing Error Model}
\label{sec:generalized_PE}
Due to the geometric spread of the optical beam, the detector can only collect a fraction of the received power, which is approximated as \cite{farid2007outage}
\begin{align}
    \label{eqn:h_p}
    h_{\text{p}, X} \approx A_0 \exp \left( - \frac{2 \bm{\rho}_X^2}{w_{L_\text{eq}}^2} \right),
\end{align}
where $\bm{\rho}_X$ is the radial displacement vector between the centers of the beam footprint and that of Bob/Eve's detector, $A_0 = \left[ \text{erf} \left( v \right) \right]^2$ is the fraction of the collected power given $\rho = 0$, $v = \left( \sqrt{\pi} r_\text{a} \right) / \left(\sqrt{2} w_{L, \text{eff}} \right)$, $r_\text{a}$ is the receiver aperture radius, $\text{erf} \left( \cdot \right)$ denotes the error function, $w_{L_\text{eq}} = \sqrt{w_{L, \text{eff}}^2 \frac{\sqrt{\pi} \text{erf} \left( v \right)}{2 v \exp \left( - v^2\right)}}$ is the equivalent beam width. Moreover, $w_{L, \text{eff}} \approx w_0 \sqrt{\left(1 - \frac{L}{F_0} \right)^2 + \left( \frac{2 L}{k_\text{wave} w_0^2}\right)^2}$ is the effective beam width, where $L = \left( H_\text{s} - H_\text{v} \right) \sec \left( \xi \right)$ is the slant distance, $H_\text{s}$ is the satellite altitude, $H_\text{v}$ is the vehicle altitude, $F_0$ represents the phase font radius of curvature, $k_\text{wave} = 2 \pi / \lambda$ denotes the wave number, $w_0 = \frac{\lambda}{\pi \theta}$ is the beam waist radius, and $\theta$ is the divergence half-angle \cite{andrews2005laser}.

\textit{As for the LEO-to-Bob channel,} the main causes for the pointing errors between the center of the satellite's optical beam on the ground and that of Bob’s detector are \textit{(1) the satellite’s vibration} and \textit{(2) the constant changes in Bob's velocity that the satellite cannot keep track of in a short period of time.} Therein, we consider the most general case, where the displacements in the $x$ and $y$ axes are two independent Gaussian random variables, with different means ($\mu_x, \mu_y$) and variances ($\sigma_x^2, \sigma_y^2$). In other words, $\bm{\rho}_\text{B} = \left( \rho_{\text{B}, x}, \rho_{\text{B}, y} \right)$, where $\rho_{\text{B}, x} \sim \mathcal{N} \left( \mu_x, \sigma_x^2 \right)$ and $\rho_{\text{B}, y} \sim \mathcal{N} \left( \mu_y, \sigma_y^2 \right)$. As a result, the radial displacement $\rho_\text{B} = \left\lVert \bm{\rho}_\text{B} \right\rVert$ follows the Beckmann distribution, which can be approximated by a modified Rayleigh distribution as \cite{boluda2016novel}
\begin{align}
    \label{eqn:approx_Beckmann_pdf}
    f_{\rho_\text{B}} \left( \rho_\text{B} \right) = \frac{\rho_\text{B}}{\sigma_\text{mod}^2}\exp\left(-\frac{\rho_\text{B}^2}{2\sigma_\text{mod}^2} \right), \quad \rho_\text{B} > 0,
\end{align}
where $\sigma_\text{mod} = \left(\frac{3\mu_x^2\sigma_x^4+3\mu_y^2\sigma_y^4+\sigma_x^6+\sigma_y^6}{2} \right)^{1/3}$ is the modified variance approximation. Using \eqref{eqn:h_p} and \eqref{eqn:approx_Beckmann_pdf}, we derive the PDF of $h_\text{p}$ as
\begin{align}
    \label{eqn:approx_h_p_pdf}
    f_{h_\text{p,B}} \left( h_\text{p,B} \right) = \frac{\varphi_\text{mod}^2}{A_\text{mod}^{\varphi_\text{mod}^2}} h_\text{p,B}^{\varphi_\text{mod}^2-1}, \quad 0 \leq h_\text{p,B} \leq A_\text{mod},
\end{align}
where $A_\text{mod} = A_0 \exp\left(\frac{1}{\varphi_\text{mod}^2}-\frac{1}{2\varphi_x^2}-\frac{1}{2\varphi_y^2}-\frac{\mu_x^2}{2\sigma_x^2\varphi_x^2}-\frac{\mu_y^2}{2\sigma_y^2\varphi_y^2}\right)$, $\varphi_\text{mod}$, $\varphi_y$, and $\varphi_x$ are the ratio between the equivalent beam waist at Bob's receiver and the jitter standard deviation, the jitter variances in the $x$ and $y$ axes, respectively. They can be written as $\varphi_\text{mod} = w_{L_\text{eq}}/2\sigma_\text{mod}$, $\varphi_x = w_{L_\text{eq}}/2\sigma_x$ and $\varphi_y = w_{L_\text{eq}}/2\sigma_y$.

\textit{For Eve's pointing error model,} we assume that Eve's detector is on the same receiver plane and $d_\text{E}$ meters from Bob's receiver. In this case, the misalignment conditions at Bob also affect Eve. Particularly, let $\textbf{P}$ denote the coordinate of the center of Bob's detector, and $\textbf{Q} = \begin{bmatrix} X' \\ Y' \end{bmatrix} + \textbf{P}$, denoting the coordinate of the center of the misaligned beam footprint, where $X'$ and $Y'$ are the beam displacements in the $x$ and $y$ axes, respectively. The distance between the center of Eve's detector and the center of the misaligned beam footprint, $\rho_\text{E}$, is given as $r_{\text{E}}^2 = \|\mathbf{Q}\|^2 = \left\| \begin{bmatrix} X' \\ Y' \end{bmatrix} \right\|^2 + 2 \begin{bmatrix} X' \\ Y' \end{bmatrix}^T \mathbf{P} + \|\mathbf{P}\|^2 = \rho_\text{B}^2 +2 \begin{bmatrix} X' \\ Y' \end{bmatrix}^T \mathbf{P} + d^2_\text{E}$.
Consequently, the pointing error coefficient of Eve can be approximated as \cite{trinh2020secrecy}
\begin{align}
     h_\text{p,E} \approx { A_{\text{mod}}}\exp \left ({ -\frac {2 \rho_\text{B}^2}{ w_{L_\text{eq}}^2 }}\right) \exp \left ({-\frac {2 d_\text{E}^2}{ w_{L_\text{eq}}^2 } }\right) \exp \left ({-U }\right),
\end{align}
where $U = \frac{4}{w_{L_\text{eq}}^2} \begin{bmatrix} X' \\ Y' \end{bmatrix}^T \textbf{P}$, which is a zero-mean normal random variable with the variable $\sigma^2_U = \frac{16 \sigma_\text{mod}^2 d_\text{E}^2}{w_{L_\text{eq}}^4}$.
\subsubsection{Composite Channel Model}
\label{sec:composite_channel_pdf}

By approximating the GG distribution as the mixture Gamma distribution, the composite PDF channel of Bob is derived as \cite{sandalidis2016tractable}
\begin{align}
    \label{eqn:composite_pdf_Bob}
    f_{h_\text{B}} \!\left( {{ h_{\text{B}}} } \right)\! =\! \chi h_\text{B}^{{ \varphi _{\text {mod}} ^{2}}-1} \!\sum _{i=1}^{n} \! a_{i} \xi_{i}^{{ \varphi _{\text {mod}}^{2}} - \alpha }\Gamma \!\left( \! {\alpha \! - \! \varphi_\text{mod}^2, \frac{\xi_i h_\text{B}}{{ A_{\text{mod}} h_\text{c}}}} \! \right),
\end{align}
where $\chi =   \varphi_\text{mod}^2 \left({{ A_{\text {mod}} h_\text{c}} }\right)^{-\varphi _{\text {mod}} ^{2}} $, $\Gamma \left( \cdot, \cdot \right)$ is the upper incomplete Gamma function, $a_i = \frac{\theta_i}{\sum_{j=1}^n \theta_j \Gamma \left( \alpha \right) \xi_j^{-\alpha}}$, $\xi_i = \frac{\alpha \beta}{x_i}$, $\theta_i = \frac{\left( \alpha \beta \right)^{\alpha} w_i x_i^{-\alpha+\beta-1}}{\Gamma \left( \alpha \right) \Gamma \left( \beta \right)}$, $n$ is the Gauss-Laguerre approximation order, $w_i$ and $x_i$ are respectively the weight factors and abscissa of the Laguerre polynomials.

Regarding Eve's channel model, based on \cite[Appendix B]{trinh2020secrecy}, we can derive its numerical approximation as
\begin{align}
    \label{eqn:composite_pdf_Eve}
    f_{{ h_{\text {E}}}} \left( {{ h_{\text {E}}} } \right)=& B_0' \sum_{k=1}^{m} C_k D_k^\beta \exp \left( {-D_k { h_{\text{E}}}} \right)  h_\text{E}^{\beta - 1},
\end{align}
where $B_0' = \sqrt{2} \sigma_{G'} B_{1}'\exp \left ({-{ \varphi_\text{mod}^2} B_2' }\right) \left[ \Gamma \left( \beta \right) \right]^{-1}$, $B_1' = 0.5\varphi_\text{mod}^2 \exp \left( 0.5\varphi_\text{mod}^4 \sigma_{G'}^2 - \varphi_\text{mod}^2 \mu_{G'} \right)$, $B_2' =\varphi_\text{mod}^2 \sigma_{G'}^2 - \mu_{G'}$, $\mu_{G'} = - 0.5 \ln \left( \frac{\alpha + 1}{\alpha} \right)$, and $\sigma_{G'}^2 = \ln \left( \frac{\alpha + 1}{\alpha} \right) + \frac{16 \sigma_\text{mod}^2 d_\text{E}^2}{w_{L_\text{eq}}^4}$, Moreover, $C_k = w_k \text{erfc} \left( x_k \right) \exp \left( x_k^2 + \sqrt{2} \sigma_{G'} \varphi_\text{mod}^2 x_k \right)$, $D_k = \beta \left[ A_\text{mod} h_\text{c} \exp \left( \sqrt{2} \sigma_{G'} x_k \! - \! \frac{2 d_\text{E}^2}{w_{L_\text{eq}}^2} \! - \! B_2' \right) \right]^{-1}$, $m$ is the Gauss-Hermite approximation order, $w_k$ and $x_k$ are the weight factors and abscissa of the Hermite polynomials, respectively.

\subsection{CV-QKD with Dual-threshold/Direct Detection (DT/DD)}
\label{sec:cv_qkd_dtdd}

As for the QKD protocol implemented in the quantum phase, we consider the non-coherent CV-QKD protocol using SIM/BPSK with a dual-threshold detector. This protocol, inspired by the BB84 protocol, offers a solution with low complexity, low implementation cost, and compatibility with conventional communication devices \cite{trinh2018design}. It is, nevertheless, important to note that the proposed scheme can also be employed for the key reconciliation of any QKD systems.

Fig.~\ref{fig:phy_layer} illustrates the block diagram of the considered system. Particularly, Alice first modulates random binary bits into SIM/BPSK signals with a small modulation depth $\delta$ $\left(0 < \delta < 1 \right)$. The modulated signals are then transmitted over the optical satellite link with a fading channel coefficient $h \left( t \right)$ and noise. At the receiver, Bob uses a PIN photodiode to directly detect the optical signals and convert them into electrical signals. The demodulated electrical signal after the matched filter is $s \left( t \right) = \left(2 i - 1 \right) \frac{1}{4} \mathfrak{R} \delta P_\text{t} h \left( t \right) + n \left( t \right)$, where $i$ denotes the transmitted bit ($ i \in \left\{ 0, 1 \right\}$), $\mathfrak{R}$ denotes the responsivity of the detector, $P_\text{t}$ is the peak transmitted power, $n \left( t \right)$ denotes the received noise. Bob detects the received signal based on dual thresholds (DTs), i.e., $d_0$ and $d_1$, as $b \left( t \right) = \begin{cases}
        0 &  \text{if } s \left( t \right) \leq d_0, \\
        1 &  \text{if } s \left( t \right) \geq d_1, \\
        \text{X} & \text{otherwise},
\end{cases}$ where $b \left( t \right)$ denotes the result of detection at time $t$, 'X' represents the case that Bob can not detect the signal. The values of DTs are set as $d_i = \left( 2 i - 1 \right) \left( \frac{1}{4} \mathfrak{R} \delta P_\text{t} h_\text{c} \frac{A_\text{mod} \varphi_\text{mod}^2}{1 + \varphi_\text{mod}^2} +  \zeta \sqrt{\sigma^2_\text{n}} \right)$ \cite{trinh2020secrecy}, where $\zeta$ denotes the DT scale coefficient, and $\sigma^2_\text{n}$ is the variance of the received noise.

Regarding the received noise $n \left( t \right)$, it can be modeled as zero-mean additive white Gaussian noise (AWGN), whose variance can be written as $\sigma^2_\text{n} = \sigma^2_\text{sh} + \sigma^2_\text{bg} + \sigma^2_\text{th}$ \cite{trinh2018design}. Herein, $\sigma^2_\text{sh}$, $\sigma^2_\text{bg}$, and $\sigma^2_\text{th}$ are the variances of the shot noise, background noise, and thermal noise, respectively. \textit{For the shot noise,} $\sigma^2_\text{sh} = 2 q \mathfrak{R} \left( \frac{1}{4} P_\text{t} \delta h \right) \Delta f$, where $\Delta f = \frac{R_\text{b}}{2}$ denotes the effective noise bandwidth, $R_\text{b}$ is the data rate of the FSO channel. \textit{Regarding the background noise,} $\sigma^2_\text{bg} = 2 q \mathfrak{R} P_\text{b} \Delta f$, where $q$ is the electron charge, $P_\text{b} = \Omega \pi r^2_\text{a} \frac{B_0 \lambda^2}{c}$ denotes the average received background radiation power, $\Omega$ represents the Sun's spectral irradiance above the atmosphere, $B_0$ denotes the optical bandwidth, $c$ is the speed of light. \textit{As for the thermal noise,} $\sigma^2_\text{th} = \frac{4 k_\text{B} T F_\text{n}}{R_\text{L}} \Delta f$, where $k_\text{B}$ denotes the Boltzmann constant, $T$ denotes the receiver temperature in Kelvin degree, $F_\text{n}$ is the amplified noise figure, and $R_\text{L}$ is the load resistance.

\section{Proposed Method}
\label{sec:kr_design}
This section first describes the proposed blind reconciliation design, considering the protograph RC-LDPC code and the syndrome-based error estimation. Then, details of the syndrome-based error estimation method are provided.

\begin{figure}[t]
    \centering
    \includegraphics[width=\linewidth]{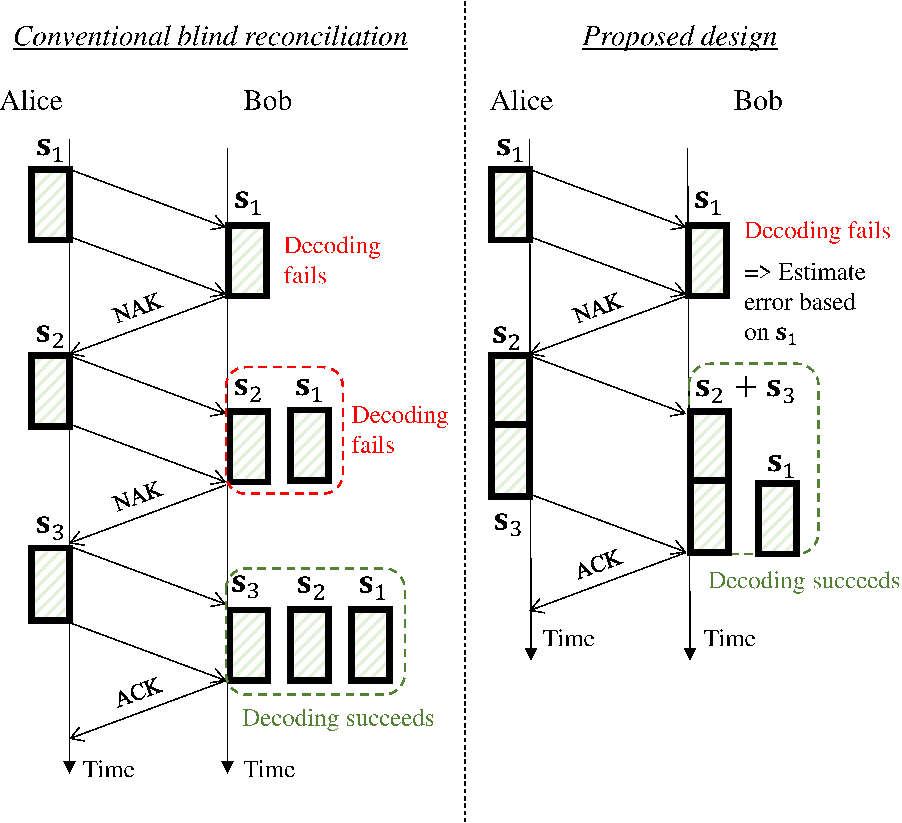}
    \caption{The operations of the conventional blind reconciliation and the proposed design.}
    \label{fig:cmp_conventional_proposal}
\end{figure}

\subsection{Blind Reconciliation Design}
\label{sec:proposed_kr_overview}

\subsubsection{Review of Blind Reconciliation}
\label{sec:rv_blind_KR}

Blind reconciliation uses the RC-LDPC code family, which allows the use of syndrome bits from high-rate codes to decode lower-rate ones. At the beginning of the procedure, both users start from the highest code rate of the family. If a reconciliation attempt fails, more syndrome bits are sent via the public channel. At the receiver, these syndrome bits are combined with the previous ones to lower the code rate, thereby improving the chances of successful reconciliation. The protocol continues until the sifted block is successfully decoded, or the parties have tried all possible code rates. In the latter case, both sides discard their sifted blocks.

However, in case the average QBER is high and the initial code rate fails most of the time, the protocol may take many communication rounds to reach the appropriate code rate. This becomes exacerbated in satellite communications, where the propagation delay is high (in the order of milliseconds). To circumvent this problem, we additionally consider the syndrome-based error estimation method \cite{lechner2013estimating}. The key idea of the proposed design is that the receiver can utilize the initial syndrome to estimate the QBER and share this information with the sender. Based on the estimated QBER, the sender can adapt the code rate by adjusting the number of syndrome bits in the next communication round. {Fig.~\ref{fig:cmp_conventional_proposal} describes the operations of the conventional blind reconciliation and the proposed design.}

\subsubsection{Protograph-based Rate-compatible (RC)-LDPC Code Family}
\label{sec:rc_ldpc_family}

To facilitate the operation of blind reconciliation, we consider the RC-LDPC code family constructed from protographs. A protograph refers to a bipartite graph with a small number of nodes that can be used as a prototype to create parity-check matrices of various sizes. This also enables the efficient design and optimization of LDPC codes at the protograph level \cite{divsalar2009capacity}.

The protomatrix of the desired RC-LDPC code family can be structured by the code extension method. In this method, the protomatrix of a low-rate code is constructed by adding new rows to that of a high-rate code and exhaustively searching for the one with the lowest decoding threshold \cite{cuong2024blind}. This results in the RC-LDPC code family, whose structure is shown in Fig.~\ref{fig:proposed_structure}. Therein, $N_\text{r}$ denotes the number of LDPC codes in the family, $1 > R_1 > R_2 > ... > R_{N_\text{r}}$ represents the code rates, and $\textbf{H}_i \left( 1 \leq i \leq N_\text{r} \right)$ denotes the corresponding parity-check matrix. An arbitrary matrix $\textbf{H}_{i+1}$ is a combination of a higher-rate matrix $\textbf{H}_{i}$ and a number of rows, forming a sub-matrix $\textbf{H}^{\Delta}_{i+1}$. This type of structure allows the construction of the syndrome vector $\textbf{s}_{i+1}$ by adding more syndrome bits to $\textbf{s}_{i}$, as illustrated in Fig.~\ref{fig:proposed_structure}.

\begin{figure}[t]
    \centering
    \includegraphics[width=\linewidth]{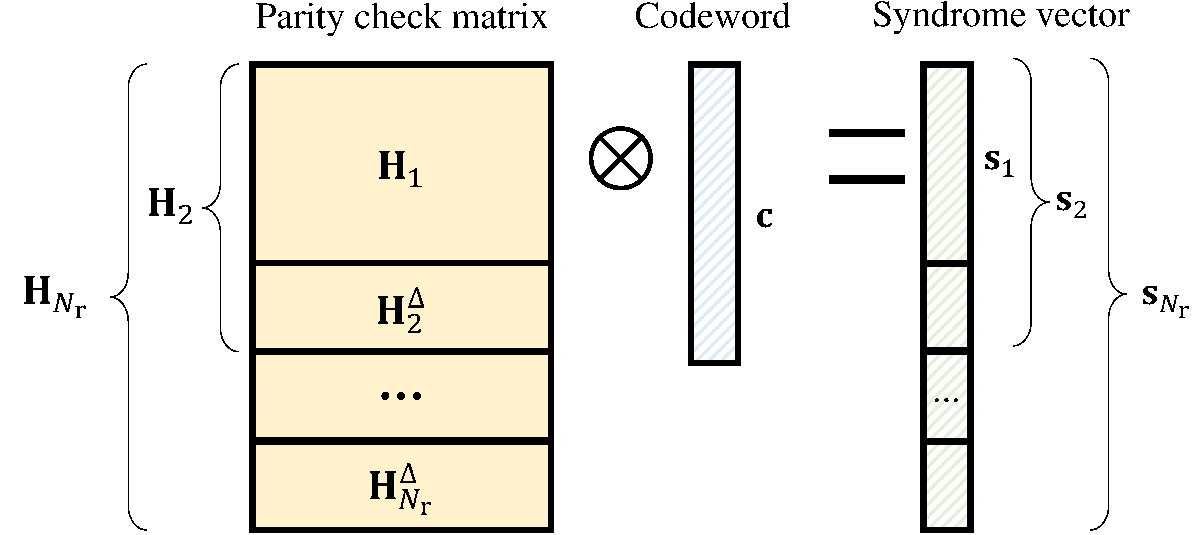}
    \caption{Illustration of the proposed RC-LDPC code family and how created syndrome vectors can be nested.}
    \label{fig:proposed_structure}
\end{figure}

\subsubsection{Operation Steps}
\label{sec:operation_step}

Let $\textbf{k}_\text{A}$ and $\textbf{k}_\text{B}$ represent the sifted blocks of Alice and Bob, respectively. 
The detailed steps of the proposed design are as follows:
\begin{itemize}
    
    \item \textbf{Step 1:} Alice creates the initial syndrome, $\textbf{s} = \textbf{k}_\text{A} \textbf{H}_1^\text{T}$, and transmits it to Bob via the public channel.

    \item \textbf{Step 2:} Bob decodes his sifted block according to syndrome $\textbf{s}$ using the modified decoding algorithm \cite{liveris2002compression}. The modified decoding algorithm operates in the same way as conventional decoding algorithms, except that it takes into account the syndrome vector $\textbf{s}$. In particular, Bob tries to decode $\textbf{k}_\text{B}$ into $\hat{\textbf{k}_\text{B}}$ so that $\hat{\textbf{k}_\text{B}} \textbf{H}_1^\text{T} = \textbf{s}$. If the sifted block is decoded successfully after a predefined number of iterations, Bob informs Alice, and they save the reconciled blocks. Otherwise, they move to the next step.

    \item \textbf{Step 3:} Bob utilizes the syndrome $\textbf{s}$ to estimate the QBER of his sifted key and select a proper code rate in the RC-LDPC code family. Details of the error estimation method are described in Section~\ref{sec:err_est_synrome}. He then shares the index of the code rate, denoted as $\epsilon$, with Alice.

    \item \textbf{Step 4:} Alice sets $i = \epsilon$ and computes the incremental syndrome $\Delta \textbf{s}$ as
    \begin{align}
        \label{eqn:delta_syndrome}
        \Delta \textbf{s} =  \textbf{k}_\text{A} \left[ \left( \textbf{H}^{\Delta}_2 \right)^\text{T} \; \dots \; \left( \textbf{H}^{\Delta}_\epsilon \right)^\text{T} \right].
    \end{align}
    
    \item \textbf{Step 5:} Bob obtains the new syndrome $\textbf{s}_\epsilon$ by concatenating the syndrome $\textbf{s}$ and $\Delta \textbf{s}$. He then tries another decoding attempt with the newly obtained syndrome. If the decoding attempt is successful, the procedure completes. Otherwise, they move to the next step.

    \item \textbf{Step 6:} Alice checks the condition $i = N_\text{r}$. If it is true, both sides discard their sifted blocks. Otherwise, Alice sets $i = i + 1$, computes a new incremental syndrome $\Delta \textbf{s} = \textbf{k}_\text{A} \left( \textbf{H}^{\Delta}_i \right)^\text{T}$, and transmits it to Bob. The procedure goes back to step 5.
\end{itemize}

\subsection{Syndrome-based Error Estimation}
\label{sec:err_est_synrome}

The considered syndrome-based QBER estimation method is based on the method proposed in \cite{lechner2013estimating} that utilizes syndromes of regular LDPC codes. This approach can be applied to protograph LDPC codes thanks to the fact that in the parity-check matrix, all rows lifted from a row in the corresponding protomatrix possess the same row weight. 

Let $\rho$ denote the bit-flipping probability between the considered sifted block of Alice and Bob, $\textbf{H}$ be an $m \times n$ parity-check matrix lifted from a protomatrix $\textbf{B}$, and $\textbf{s}_\text{A}$ and $\textbf{s}_\text{B}$ be the syndromes created from $\textbf{H}$ and the sifted blocks of Alice and Bob, respectively. We denote $\textbf{s}_{\text{A}, d} \subseteq \textbf{s}_\text{A}$ and $\textbf{s}_{\text{B}, d} \subseteq \textbf{s}_\text{B}$ as the vector of syndrome bits created from rows with the same weight $d$. According to Remark~\ref{remark:protograph_weight}, the length of $\textbf{s}_{\text{A}, d}$ and $\textbf{s}_{\text{B}, d}$, denoted as $m_d$, have the minimum value equaling the lifting factor, $n_\text{sift}$.
\begin{remark}
    \label{remark:protograph_weight}
    Let $\textbf{w}$ be an arbitrary row in the protomatrix $\textbf{B}$, whose weight is $d$. In the parity-check matrix $\textbf{H}$, there are at least $n_\text{\rm lift}$ rows that have the weight of $d$, where $n_\text{\rm lift}$ is the lifting factor.
\end{remark}

\begin{figure}
    \centering
    \includegraphics[width=7.6cm]{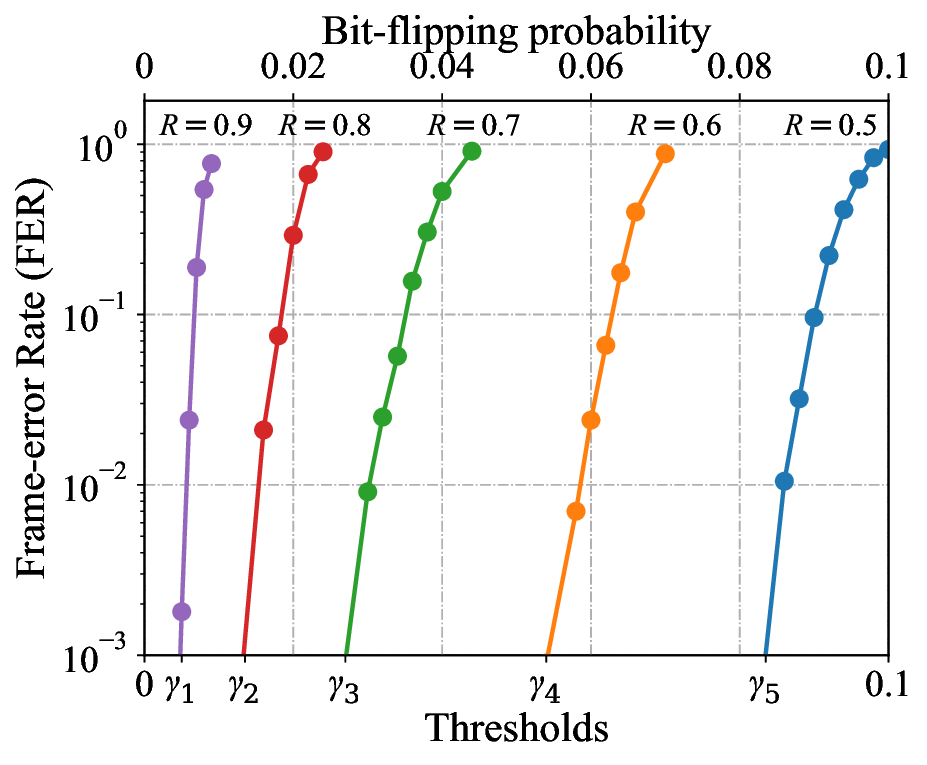}
    \caption{FER performance of selected code rates with different values of bit-flipping probabilities (upper y-axis). The thresholds for code rate determination are shown on the lower y-axis, where $\left\{ \gamma_1, \gamma_2, \gamma_3, \gamma_4, \gamma_5 \right\} = \left\{ 0.005, 0.0135, 0.027, 0.054, 0.0835 \right\}$.}
    \label{fig:thresholds_10k}
\end{figure}

Furthermore, let $ 0 < \gamma_2 < \cdots < \gamma_\text{th} \leq 0.5$ represent the thresholds that Alice uses to determine the code rate. In this work, we select the thresholds so that the frame-error rate (FER) is below a predefined level, i.e., $10^{-3}$. For analytical purposes, this paper determines the thresholds based on the empirical performance of LDPC codes. For instance, let's consider the FER performance of an RC-LDPC code family with the code members' rates $\{0.9, 0.8, 0.7, 0.6, 0.5\}$, which is shown in Fig.~\ref{fig:thresholds_10k}. The parity-check matrix of this family is created by lifting its protomatrix, $\textbf{B}_{1/2}$, twice, resulting in the final block length of $10,000$ bits. Details of the two-step lifting and the construction of $\textbf{B}_{1/2}$ can be found in \cite{cuong2024toward} and Appendix~\ref{app:code_design_example}, respectively. As seen, the thresholds of each code member are selected as $\left\{ \gamma_1, \gamma_2, \gamma_3, \gamma_4, \gamma_5 \right\} = \left\{ 0.005, 0.0135, 0.027, 0.054, 0.0835 \right\}$, corresponding to the points where their FER plots are $10^{-3}$.

Details of the error estimation method and code rate selection are as follows:
\begin{itemize}
    
    \item \textbf{Step 1:} Bob selects the minimum row weight value among all rows of protomatrix $\textbf{B}$, $d_\text{min}$\footnote{This is because the higher value of $d$ leads to the higher mean squared error of the estimator \cite{lechner2013estimating}.}, and creates two vectors, i.e., $\textbf{s}_{\text{A}, d_\text{min}}$ and $\textbf{s}_{\text{B}, d_\text{min}}$ based on $\textbf{s}_\text{A}$ and $\textbf{s}_\text{B}$, respectively. Then he computes the difference between the two vectors as $\Delta \textbf{s} = \textbf{s}_{\text{A}, d_\text{min}} \oplus \textbf{s}_{\text{B}, d_\text{min}}$.

   \item \textbf{Step 2:} Bob estimates the probability that a syndrome is flipped, $q$, as $\hat{q} = \frac{w}{m_{d_\text{min}}}$,     where $w$ is the weight of $\Delta \textbf{s}$.

    \item \textbf{Step 3:}  Given $\rho \in \left[ 0, 0.5 \right]$, the estimated value of $\rho$, denoted as $\hat{\rho}$, is expressed as \cite{lechner2013estimating}
    \begin{align}
    \hat{\rho} =  
        \begin{dcases}
            \phi^{-1} \! \left( \! w \! \right) = \frac{1 \! - \! \left( \! 1 \! - \! 2  \frac{w}{m_{d_\text{min}}} \right)^\frac{1}{d}}{2}, & \text{if } \hat{q} < 1/2,\\
            \frac{1}{2}, &\text{if } \hat{q} \geq 1/2,
        \end{dcases}
    \end{align}
    where $\phi \left( \rho \right) = q$ is the theoretical bit-flipping probability of the syndrome and can be modeled as a Bernoulli process \cite{lechner2013estimating, kiktenko2018error}. Therefore, the expression of $\phi \left( \cdot \right)$ is written as
    \begin{align}
        \label{eqn:syn_bit_one}
        \phi \left( \rho \right) = \!\!\!\!\!\! \displaystyle\sum^{d}_{\substack{i=1 \\
                                            i \equiv 1 \left( \mathrm{mod} 2 \right) }}
        \!\!\!\!\!\!\left(
        \begin{array}{c}
          d \\
          i
        \end{array}
      \right) \rho^i \left(\! 1 \! - \!\rho\! \right)^{d - i} = \frac{1 \!- \! \left( 1 \! - \! 2 \rho \right)^d}{2}.
    \end{align}

    \item \textbf{Step 4:} Bob selects the smallest integer $\epsilon \in \left[ 2, N_\text{r} \right]$ so that $\hat{\rho} \leq \gamma_\epsilon$. If no value is satisfied, he sets $\epsilon = N_\text{r}$. Finally, he announces to Alice the value of $\epsilon$ and finishes the procedure.
\end{itemize}

Our proposed design offers three key advantages in implementation complexity and practical deployment compared to existing schemes. \textit{First,} existing blind reconciliation designs \cite{martinez2012blind, kiktenko2017symmetric, liu2020blind, mao2021high, borisov2022asymmetric} rely on unstructured LDPC codes. Their highly random parity-check matrices lead to high hardware implementation complexity. In contrast, the parity-check matrix of protograph LDPC codes inherits the structure of the underlying protograph, enabling highly efficient hardware decoder implementation \cite{divsalar2009capacity}. \textit{Secondly,} most existing works utilize shortening and puncturing techniques to construct RC-LDPC codes, which restricts their cover range due to constraints on the number of auxiliary bits \cite{kiktenko2021blind}. Conversely, code members in a protograph RC-LDPC family are designed and optimized at the protograph level. This allows a single code family to cover a wide range of QBER \cite{tarable2024rateless, cuong2024blind}. \textit{Finally,} our error estimation method is considerably simpler than those for irregular LDPC codes \cite{kiktenko2018error, gao2019multi}, which require exhaustive searches. By utilizing protograph LDPC codes with fixed row weights for a given number of rows, our approach substantially reduces computational complexity.

\section{Performance Analysis}
\label{sec:performance_analysis}
This section first provides the mathematical framework for the conventional blind reconciliation, i.e., blind reconciliation without error estimation, in terms of the end-to-end SKT. Based on that, we develop the framework for the proposed design, accounting for the effects of imperfect error estimation. Finally, the statistical distribution of QBER is investigated to complete the framework.

\subsection{Conventional Blind Reconciliation}
\label{sec:final_key_rate}
In this section, we consider the end-to-end SKT, which is defined as the average number of secret bits produced by the system per second. The SKT of the conventional blind reconciliation can be expressed as
\begin{align}
    \label{eqn:skr}
    \text{SKT} = \frac{R_\text{sec}}{T_\text{sf}} = \frac{N_\text{b} n_\text{sift} \displaystyle\sum_{i=1}^{N_\text{r}} P^{\left( i \right)}_{\text{succ}} \left( \beta_i I_\text{AB} - I_\text{E} \right)}{\overline{\varepsilon}_\text{Q} + \left( N_\text{b} - 1 \right) \text{max} \left[ \overline{\varepsilon}_\text{Q}, \overline{\varepsilon}_\text{P} \right] + \overline{\varepsilon}_\text{P}},
\end{align}
where $R_\text{sec}$ is the average number of secret bits per superframe, $T_\text{sf}$ is the average duration of a superframe, $P^{\left( i \right)}_{\text{succ}}$ is the percentage of simulated frames corrected by $i$-th code rate and can be calculated as in \eqref{eqn:P_succ},
$n_\text{sift}$ is the length of a sifted block in bits, $N_\text{r}$ denotes the maximum number of code rates in the family, $I_\text{AB}$ represents the mutual information between the sifted key of Alice and that of Bob, $\beta_i = \frac{C_i}{I_\text{AB}}$ is the reconciliation efficiency, $C_i$ is the $i$-th code rate, and $I_\text{E}$ denotes the measure of the information that the eavesdropper can obtain over the quantum channel, whose calculation is shown in the next paragraph.
Furthermore, $\overline{\varepsilon}_\text{Q}$ is the average time to share a sifted key over the quantum channel and given in \eqref{eqn:avg_t_qp}, and $\overline{\varepsilon}_\text{P}$ is the average time to process a sifted key over the public channel and given in \eqref{eqn:avg_t_pp}.

In this paper, we consider the URA, where Eve attempts to gain information about the secret key by eavesdropping and detecting the signals from Alice using a detection threshold $d_\text{E, th} = 0$. As a result, Eve derives her own sifted key, which is partially correlated with Alice's. This procedure can be equivalently viewed as the data transmission over the binary symmetric channel. Therefore, the mutual information between the sifted key of Alice and that of Eve is $I_\text{E} = 1 - H \left( p_\text{e} \right)$, where $H \left( \cdot \right)$ is the binary entropy function, $p_\text{e}$ is the probability that Eve detects Alice’s transmitted bits incorrectly \cite{trinh2020secrecy}. Moreover, $\overline{\varepsilon}_\text{Q}$ can be expressed as
\begin{align}
    \label{eqn:avg_t_qp}
    \overline{\varepsilon}_\text{Q} = \frac{n_\text{sift}}{R^\text{FSO}_\text{b} P_\text{sift}},
\end{align}
where $R^\text{FSO}_\text{b}$ is the data rate of the FSO channel, and $P_\text{sift}$ denotes the probability that Bob can detect bits '0' and '1' according to the rule in Sec.~\ref{sec:cv_qkd_dtdd}. The computation of $p_\text{e}$ and $P_\text{sift}$ can be obtained via Lemma~\ref{LEM:P_SIFT_QBER}.

\begin{figure*}[!t]
\normalsize
    \begin{align}
        \label{eqn:closed_form_P_sift}
        P_\text{sift} \approx \varphi_\text{mod}^2 
        \displaystyle\sum_{a, b \in \left\{ 0, 1 \right\} }
        \displaystyle\sum^{n'}_{j=1}
        \displaystyle\sum_{i=1}^{n} w_j \exp \left( x_j \right) Q \left[ \frac{ \Upsilon' \left(a, b, x_j \right) }{ \sigma_{\text{n}, k} \left( x_j \right) } \right]  x_j^{ \varphi_\text{mod}^2 - 1}  a_{i} \xi_{i}^{{ \varphi _{\text {mod}}^{2}} - \alpha }\Gamma \left( {\alpha -  \varphi_\text{mod}^2, 
        \xi_i x_j
        } \right).
    \end{align}

    \begin{align}
        \label{eqn:closed_form_ber_eve}
        p_\text{e} \approx B_0' \left( h_\text{c} A_\text{mod}\right)^{\beta} \displaystyle\sum_{i=1}^{n} \displaystyle\sum_{k=1}^{m} w_i \exp \left( x_i \right) Q \left[ \frac{ 0.25 \mathfrak{R} \delta P_\text{t} x_i h_\text{c} A_\text{mod}}{ \sigma_{\text{n}, k} \left( x_i \right) } \right]   C_k D_k^\beta \exp \left( {-D_k x_i h_\text{c} A_\text{mod} } \right)  x_i^{\beta - 1}.
    \end{align}
\hrulefill
\vspace*{4pt}
\end{figure*}

\begin{lemma}
\label{LEM:P_SIFT_QBER}
     The average sift probability at Bob, $P_\text{sift}$, and the average error probability at Eve, $p_\text{e}$ can be approximated as in \eqref{eqn:closed_form_P_sift} and \eqref{eqn:closed_form_ber_eve}, respectively.
     In this regard, $Q \left( \cdot \right)$ is the Q-function, $\sigma_{\text{n}, k} \left( x_k \right) = \sqrt{2 q \mathfrak{R} \left( \frac{1}{4} P_\text{t} \delta h_\text{c} A_\text{mod} x_k + P_\text{b} \right) \Delta f + \sigma^2_\text{th} }$, $\Upsilon' \left(a, b, x_k \right) = \left( 1 - 2 b \right) i_{a, k} + \left( 2 b - 1 \right) d_{b, k}$, $i_{a, k} = \frac{\left( 2a - 1 \right)}{4} \mathfrak{R} \delta P_\text{t} x_k h_\text{c} A_\text{mod}$, $d_{b, k} =  \left( 2 b - 1 \right) \left( \frac{1}{4} \mathfrak{R} \delta P_\text{t} h_\text{c} \frac{A_\text{mod} \varphi_\text{mod}^2}{1 + \varphi_\text{mod}^2} +  \zeta \sqrt{\sigma^2_{\text{n}, k}} \right)$, $n'$ is the Gauss-Laguerre approximation order, $w_j$ and $x_j$ are the weight factors and abscissa of the Laguerre polynomials, respectively.
\end{lemma}

\begin{proof}
    Please see Appendix~\ref{sec:proof_p_sift_qber}
\end{proof}

\begin{remark}
    The approximation \eqref{eqn:closed_form_P_sift} quickly converges to the exact values with the approximation order of $n' = 25$ and $n = 20$. For \eqref{eqn:closed_form_ber_eve}, the approximation orders needed for convergence are around $n = m = 50$.
\end{remark}

Furthermore, $\overline{\varepsilon}_\text{P}$ can be computed as
\begin{align}
    \label{eqn:avg_t_pp}
    \overline{\varepsilon}_\text{P} = t_\text{prop} + t^\text{sifting}_\text{trans} + \displaystyle\sum_{i=1}^{N_\text{r}} P^{\left( i \right)}_{\text{succ}} D_i + P_{\text{fail}} D_{N_\text{r}},
\end{align}
where $t_\text{prop} = \frac{L}{c}$ is the propagation delay, $t^\text{sifting}_\text{trans}$ denotes the transmission delay of sifting information, and can be calculated as
\begin{align}
    \label{eqn:delay_trans_sifting}
    t^\text{sifting}_\text{trans} = \frac{n_\text{sp}}{R^\text{RF}_\text{b}}.
\end{align}
Therein, $n_\text{sp}$ is the length of the sifting positions information, and $R^\text{RF}_\text{b}$ is the data rate of the public channel. Moreover, $D_i$ is the average delay in case the reconciliation finishes in the $i$-th communication round and is expressed as
\begin{align}
    \label{eqn:avg_delay}
    D_i = \left( 2 i - 1 \right) t_\text{prop} + \displaystyle\sum^{i}_{j=1} t^{\text{SYN}_j}_\text{trans},
\end{align}
where $t^{\text{SYN}_j}_\text{trans} = \frac{n_{\text{syn}, j}}{R^\text{RF}_\text{b}}$ is the transmission delay of the $j$-th syndrome, $n_{\text{syn}, j}$ is the length of the $j$-th syndrome, and $P_{\text{fail}} = 1 - \sum_{i=1}^{N_\text{r}} P^{\left( i \right)}_{\text{succ}}$ is the probability of fail reconciliation. Finally, $P^{\left( i \right)}_{\text{succ}}$ can be computed as
\begin{align}
    \label{eqn:P_succ}
    P^{\left( i \right)}_{\text{succ}} = F_\rho \left( \gamma_i \right) - F_\rho \left( \gamma_{i-1} \right),
\end{align}
where $ \left( \gamma_0 = 0 \right) < \gamma_1 < \cdots < \gamma_\text{th} \leq 0.5$ are the thresholds for code rate determination. These thresholds are defined based on the empirical performance of members in the RC-LDPC code family, as shown in Fig.~\ref{fig:thresholds_10k}.

In addition, $F_\rho \left( \cdot \right)$ is the CDF of QBER of sifted blocks. However, it is worth noting that finding a closed-form expression for this function is not straightforward due to the fact that the obtained sifted bits and erroneous bits in each timeslot depend on the instantaneous channel fading coefficient. As a result, each finite-size sifted block may be formed from different numbers of coherence times with different error patterns. For the sake of system analysis, we will model the statistical distribution of the QBER of sifted blocks as an exponentiated Weibull (EB) random variable. The parameters of that random variable are found by curve-fitting the histogram of the simulation data. Details of the method and the comparisons with other statistical distributions are presented in Section~\ref {sec:dis_qber}. Consequently, the CDF of QBER of sifted blocks in Eq.~\eqref{eqn:P_succ} can be substituted by that of the EB distribution, $F_\text{EW} \left( x \right)$, which is given as in \cite{barrios2013exponentiated}.

\subsection{The Proposed Design with Imperfect Error Estimation}
\label{sec:imperfect_est}

\begin{figure}[ht!]
    \centering
    \includegraphics[width=.8\linewidth]{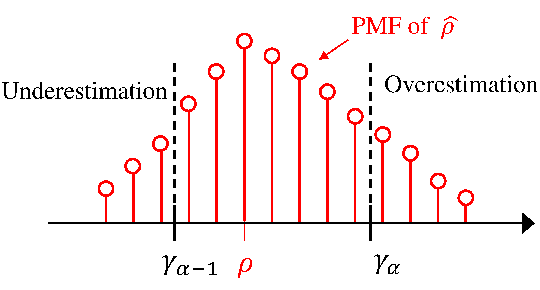}
    \caption{Illustration of underestimation and overestimation in the considered system.}
    \label{fig:under_over_est}
\end{figure}

Next, we analyze the SKT performance of the proposed design, accounting for imperfect error estimation. The imperfection happens due to the finite length of the syndrome, making the estimated values, $\hat{\rho}$, a discrete random variable. This can result in a wrong determination of state, as depicted in Fig.~\ref{fig:under_over_est}, and thus affect system performance. Particularly, let $\alpha$ denote the state of actual value $\rho$, and $\epsilon$ denote the state of the estimate $\hat{\rho}$. If $\alpha < \epsilon$, $\hat{\rho}$ is underestimated, and Bob will need more communication rounds to reach the proper code rate. On the other hand, if $\alpha > \epsilon$, $\hat{\rho}$ is overestimated, and the reconciliation step will leak more information than the necessary amount.

In such cases, the SKT can be computed as
\begin{align}
    \label{eqn:skr_imperfect}
    \text{SKT}_\text{imp} = \frac{R_\text{imp, sec} }{\overline{\varepsilon}_\text{Q} + \left( N_\text{b} - 1 \right) \text{max} \left[ \overline{\varepsilon}_\text{Q}, \overline{\varepsilon}_{\text{imp, P}} \right] +\overline{\varepsilon}_{\text{imp, P}}},
\end{align}
where $R_\text{imp, sec}$ is the average number of secret bits per superframe in the case of imperfect estimation and is expressed as
\begin{align}
    \label{eqn:r_sec_imperfect}
    R_\text{imp, sec} = N_\text{b} n_\text{sift} \displaystyle\sum_{\alpha=1}^{N_\text{r}} \displaystyle\sum_{\epsilon=1}^{N_\text{r}+1} P^{\text{est}}_{\alpha, \epsilon} \left( \beta_\chi I_\text{AB} - I_\text{E} \right),
\end{align}
where $P^{\text{est}}_{\alpha, \epsilon}$ denotes the probability that the actual value $\rho$ in state $\alpha$ and the estimate $\hat{\rho}$ is in state $\epsilon$, whose calculations is shown in \eqref{eqn:prob_err_est_1}, \eqref{eqn:prob_err_est_2}, $\chi = \max \left( \alpha, \chi_\epsilon \right)$, and $\chi_\epsilon = \min \left( \epsilon, N_\text{r} \right)$. As for the average time of the post-processing phase of a sifted block, it is calculated as
\begin{align}
    \label{eqn:avg_pp_imperfect}
    \overline{\varepsilon}_{\text{imp, P}} = t_\text{prop} + t^\text{sifting}_\text{trans} + \sum_{\alpha=1}^{N_\text{r}+1} \sum_{\epsilon=1}^{N_\text{r}+1} P^{\text{est}}_{\alpha, \epsilon} D_{\alpha, \epsilon},
\end{align}
where $D_{\alpha, \epsilon}$ is the average delay when the actual value $\rho$ in state $\alpha$ and the estimate $\hat{\rho}$ is in state $\epsilon$. We can calculate it as
\begin{align}
    \label{eqn:delay_imperfect}
    D_{\alpha, \epsilon} = 
    \begin{dcases}
        \left[ 2 \min \left( \alpha, 2 \right) - 1 \right] t_\text{prop} + \displaystyle\sum^{\chi_\epsilon}_{j=1} t^{\text{SYN}_j}_\text{trans}, & \epsilon \geq \alpha,\\
        \left[ 3 + 2 \left( \chi_\alpha - \epsilon \right) \right] t_\text{prop} + \displaystyle\sum^{\chi_\alpha}_{j=1} t^{\text{SYN}_j}_\text{trans}, & \epsilon < \alpha,\\
    \end{dcases}
\end{align}
where $\chi_\alpha = \min \left( \alpha, N_\text{r} \right)$.

To compute $P^{\text{est}}_{\alpha, \epsilon}$, we consider two possible cases as follows. First, Bob successfully reconciles the sifted block using the initial syndrome $\textbf{s}$; second, he fails the first attempt and estimates the QBER using the initial syndrome $\textbf{s}$. As for the first case, the computation of $P^{\text{est}}_{1, \epsilon}$ can be given as
\begin{align}
    \label{eqn:prob_err_est_1}
    P^{\text{est}}_{1, \epsilon} =
    \begin{dcases}
        F_\text{EW} \left( \gamma_1 \right), & \epsilon = 1,\\
        0, & \epsilon \neq 1.
    \end{dcases}
\end{align}
Regarding the second case $\left( \alpha \neq 1 \right)$, the expression of $P^{\text{est}}_{\alpha, \epsilon}$ is
\begin{align}
    \label{eqn:prob_err_est_2}
    P^{\text{est}}_{\alpha, \epsilon} =
    \begin{dcases}
            0, & \epsilon = 1,\\
            P^{\gamma_{2}}_\alpha, & \epsilon = 2,\\
            P^{0.5}_\alpha - P^{\gamma_{N_\text{r}}}_\alpha + P^{\text{fail}}_\alpha, & \epsilon = N_\text{r},\\
            P^{\gamma_{\epsilon}}_\alpha - P^{\gamma_{\epsilon-1}}_\alpha, & \text{otherwise},
        \end{dcases}
\end{align}
where $P^{\gamma}_\alpha$ is the probability that the estimate $\hat{\rho}$ is smaller than an arbitrary value $\gamma$ given the actual values $\rho$ in the state $\alpha$, or $P^{\gamma}_\alpha = P \left( \hat{\rho} \leq \gamma \;\middle|\; \rho \text{ in state $\alpha$} \right)$, and $P^{\text{fail}}_\alpha$ is the probability that Bob can not estimate because $\frac{w}{m} > 1/2$ given the actual values $\rho$ in the state $\alpha$ and given in \eqref{eqn:prob_fail_est}. For $P^{\gamma}_\alpha$, it is given as
\begin{align}
    \label{eqn:prob_underest}
    P^{\gamma}_\alpha
    = \int\limits_{\gamma_a}^{\gamma_{a+1}} f_\text{EW} \left( \rho \right) P \left( \hat{\rho} \leq \gamma \;\middle|\; \rho \right) \mathrm{d} \rho,
\end{align}
where $f_\text{EW} \left( \cdot \right)$ is the PDF of the EB distribution \cite{barrios2013exponentiated}. The probability in the last identity is derived from Lemma~\ref{LEM:CDF_P}.  
\begin{lemma}
    \label{LEM:CDF_P}
    The closed-form expression for the probability that the estimate $\hat{\rho}$ is smaller than an arbitrary value $\gamma \in \left[ 0, 0.5 \right]$ given an actual value $\rho$ is
    \begin{align}
    \label{eqn:cdf_p_final}
    P \left( \hat{\rho} \leq \gamma \;\middle|\; \rho \right) = I_{1 - \phi \left( \rho \right)} \left( m - \phi^{-1} \left( \gamma \right), \phi^{-1} \left( \gamma \right) + 1 \right),
\end{align}
where $I_\cdot \left( \cdot, \cdot \right)$ is the regularized incomplete beta function.
\end{lemma}
\begin{proof}
Please see Appendix~\ref{sec:proof_cdf_p}.
\end{proof}
Due to the analytical intractability of integration, we derive the numerical expression of \eqref{eqn:prob_underest} by first changing the interval to $\left[ -1, 1 \right]$ and then applying the Gauss-Legendre quadrature method \cite{abramowitz1988handbook}. As a result, the equation \eqref{eqn:prob_underest} can be approximated as 
\begin{align}
    \label{eqn:closed_form_prob_underest}
    P^{\gamma}_\alpha \!\! \approx \! \psi_\alpha \! \displaystyle\sum_{l=1}^\ell \! w_l f_\text{EW} \! \left( \! \Xi_l \! \right) \! I_{1 - \phi \left( \Xi_l \right) }  \!\left[ \! m \! - \! \phi^{-1} \! \left( \! \gamma \! \right), \phi^{-1} \! \left( \! \gamma \! \right) \! + \! 1 \! \right],
\end{align}
where $\psi_\alpha = \frac{\gamma_{a+1} - \gamma_{a}}{2}$, $\Xi_l = \frac{\gamma_{a+1} - \gamma_{a}}{2} x_l + \frac{\gamma_{a} + \gamma_{a + 1}}{2}$, $x_l$ is the $l$-th root of Legendre polynomials $P_n \left( x \right)$, $w_l = \frac{2}{\left( 1 - x_l \right)^2 \left[ P'_n \left( x_l \right) \right]}$, and $\ell$ is the Gauss-Laguerre approximation order.

Capitalizing on Lemma~\ref{LEM:CDF_P}, $P^{\text{fail}}_\alpha$ is expressed by
\begin{align}
    \label{eqn:prob_fail_est}
    P^{\text{fail}}_\alpha & = \int\limits_{\gamma_a}^{\gamma_{a+1}} f_\text{EW} \left( \rho \right) \text{Pr} \left[ W > \frac{m}{2} \right] \mathrm{d} p \nonumber \\
    & = \int\limits_{\gamma_a}^{\gamma_{a+1}} f_\text{EW} \left( \rho \right) \left[ 1 - I_{1 - \phi \left( p \right)} \left( \frac{m}{2}, \frac{m}{2} + 1 \right) \right] \mathrm{d} p.
\end{align}
By applying the same mathematical manipulations used with Eq.~\eqref{eqn:prob_underest}, we obtain the numerical expression of $P^{\text{fail}}_\alpha$ as 
\begin{align}
    \label{eqn:closed_form_prob_fail_est}
    P^{\text{fail}}_\alpha \! \approx \psi_\alpha \displaystyle\sum_{l=1}^\ell w_i f_\text{EW} \left( \Xi_l \right) \left[ 1 - I_{1 - \phi \left( \Xi_l \right)} \left( \! \frac{m}{2}, \! \frac{m}{2} \! + \! 1 \! \right) \! \right].
\end{align}

\begin{remark}
Approximations \eqref{eqn:closed_form_prob_underest} and \eqref{eqn:closed_form_prob_fail_est} quickly converge to exact forms, i.e., \eqref{eqn:prob_underest} and \eqref{eqn:prob_fail_est}, respectively, after $\ell=10$.
\end{remark}

\begin{figure*}[t]
    \centering
    \includegraphics[width=.9\textwidth]{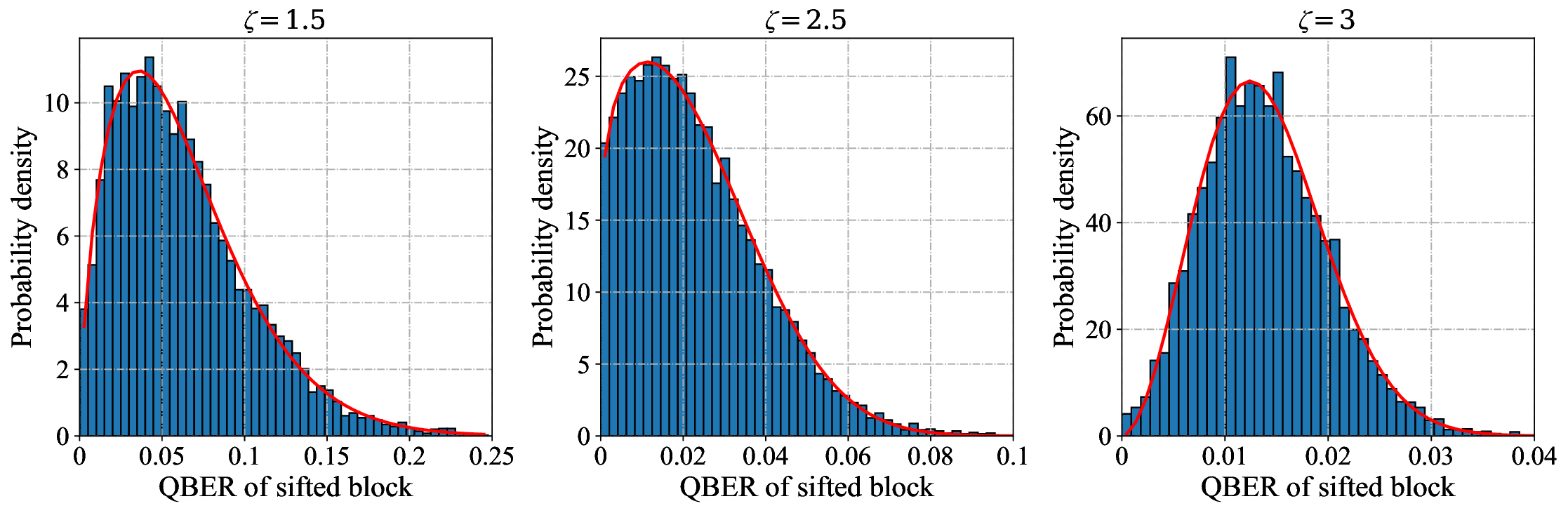}
    \caption{Examples of simulation data histograms for sifted blocks' QBER and the fitting curves of the EB distribution with different DT scale coefficients, $\zeta$.}
    \label{fig:test_results}
\end{figure*}

\begin{table*}[t]
\centering
\caption{GoF results and fitting parameters of considered probability distributions for different parameter settings}
\resizebox{\textwidth}{!}{%
\begin{tabular}{| M{3.5cm} | M{2.8cm} | M{2.8cm} | M{3cm} | M{2.8cm} | M{2.8cm} }
\hline
\multirow{2}{*}{\textbf{Parameters}} & \textbf{Normal}  & \textbf{Log-normal} & \textbf{Weibull} & \textbf{Exponentiated Weibull} \\ 
\cline{2-5}
 & GoF, ($\mu_\text{N}$, $\sigma_\text{N}$)  & GoF, ($\mu_\text{LN}$, $\sigma_\text{LN}$)  & GoF, ($\beta_\text{W}$, $\eta_\text{W}$)  & GoF, ($\alpha_\text{EW}$, $\beta_\text{EW}$, $\eta_\text{EW}$) \\
\hhline{|=|=|=|=|=|}
$\zeta = 1.5$ & $0.9484$, ($0.047, 0.03883$) & $0.937$, ($-2.8467, 0.8079$) & $0.99$, ($1.571$, $0.069$)  & $0.9903$, ($1.1147, 1.4672, 0.0657$) \\  \hline
$\zeta = 2$ & $0.9512$, ($0.0222, 0.0247$) & $0.9672$, ($-3.4325, 1.033$) & $0.9972$, ($1.3143, 0.0409$)  & $0.9983$, ($1.3246, 1.1006, 0.0342$) \\  \hline
$\zeta = 2.25$ & $0.9499$, ($0.0179,  0.02$) & $0.9192$, ($-3.62, 1.1027$) & $0.9955$, ($1.2744$, $0.0331$) & $0.9974$, ($0.744, 1.5612, 0.0386$) \\ \hline
$\zeta = 2.5$ & $0.9564$, ($0.0166, 0.0161$) & $0.8603$, ($-3.793, 0.9559$) & $0.9849$, ($1.371$, $0.027$) & $0.9969$, ($0.5412, 2.1262, 0.0348$) \\ \hline
$\zeta = 2.75$ & $0.9867$, ($1.8442, 0.0206$) & $0.8996$, ($-4.0149, 0.6545$) & $0.9799$ ($2.5377, 0.0155$) & $0.9859$, ($0.6247, 2.5, 0.024$) \\ \hline
$\zeta = 3$ & $0.9831$, ($0.0131, 0.006$) & $0.9585$, ($-4.2743, 0.4644$) & $0.989$ ($2.5377, 0.0155$) & $0.9902$, ($1.3141, 2.171, 0.0142$) \\ \hline
\end{tabular}
}
\label{tab:test_results}
\end{table*}

\subsection{Statistical Distributions of QBER of Sifted Blocks}
\label{sec:dis_qber}

This section aims to derive the statistical distribution of sifted blocks' QBER by the curve-fitting method. In particular, we find a probability distribution whose PDF best fits the histogram of the simulation data. To assess the fitness of a distribution with simulation data, we adopt the well-known $R^2$ ($R$-square) measure \cite{le2021cloud}. It is clear that the closer the $R^2$ measure is to 1, the better the fit of the predicted probability distribution to the data. 

Table~\ref{tab:test_results} presents the GoF results and fitting parameters of the considered probability distributions, i.e., normal, log-normal, Weibull, and EB distributions, for scenarios with different values of the DT scale coefficient $\zeta$. Therein, $\mu_\text{N}$ and $\sigma_\text{N}$ denote the mean and variance of the normal distribution, respectively; $\mu_\text{LN}$ and $\sigma_\text{LN}$ are the location parameter and scale parameter of the log-normal distribution; $\beta_\text{W}$ is the shape parameter, and $\eta_\text{W}$ is the scale parameter of the Weibull distribution; $\alpha_\text{EW}$ and $\beta_\text{EW}$ represent the shape parameters, and $\eta_\text{EW}$ is the scale parameter of the EB distribution. Other parameters used in the simulation can be found in Section~\ref{sec:results_discuss}.

For each scenario, we generate $10^4$ samples\footnote{It should be noted that such an amount of samples can be collected over a relatively short time frame (a few milliseconds), in which the zenith angle can be reasonably assumed to remain constant.} and create its histogram with $50$ bins. The histogram is then fitted with the PDF of the considered probability distributions using the \texttt{curve\_fit} function from the open-source SciPy library. As seen, the EB distribution provides the best fit among the considered statistical distributions, with $R^2 > 0.99$ in most scenarios. This is because the EB distribution has three parameters, providing a higher degree of freedom to fit the random QBER values across sifted blocks.
Consequently, for the rest of the paper, the EB distribution is considered for modeling the statistical distribution of sifted blocks' QBER.

\section{Numerical Results \& Discussions}
\label{sec:results_discuss}

\subsection{Parameter Settings}
\label{sec:para_setting}

Unless otherwise stated, the parameters considered for the analysis are as follows. {\textit{At the LEO satellite (Alice):} satellite altitude $H_\text{s} = 500$ km, zenith angle $\xi = 30^{\circ}$, divergence half-angle $\theta = 25$ $\mu\text{rad}$, $\mu_x = \mu_y = 0$, jitter angles $\theta_{\text{jt}, x} = \theta_{\text{jt}, y} = 2.5$ $\mu\text{rad}$, data rate of the FSO channel $R^\text{FSO}_\text{b} = 1$ Gbps, data rate of the public channel $R^\text{Pub}_\text{b} = 1$ Gbps, modulation depth $\delta = 0.2$, $F_0 = \infty$, transmitted power $P_\text{t} = 25$ dBm, optical wavelength $\lambda = 1.55$ $\mu$m, number of blocks of sifted key per superframe $N_\text{b} = 100$.} \textit{At the ground vehicle (Bob and Eve):} the vehicle altitude $H_\text{v} = 1.5$ m, receiver aperture radius $r_\text{a} = 5$ cm, detector responsivity $\mathfrak{R} = 0.9$ A/W, load resistor $R_\text{L} = 1$ $k\Omega$, amplified noise figure $F_\text{n} = 2$, receiver temperature $T = 298$ K, optical bandwidth $B_0 = 250$ GHz, the DT scale coefficient at Bob $\zeta = 2.5$, and the minimum distance between Bob and Eve $d_\text{E} = 10$ m. \textit{Regarding the FSO channel:} atmospheric altitude $H_\text{a} = 20$ km, rms wind speed $v_\text{wind} = 21$ m/s, ground turbulence level $C^2_\text{n} \left( 0 \right) = 10^{-14}$ $\text{m}^{-2/3}$, channel coherence time $t_\text{coh} = 1$ ms, the Sun's spectral irradiance $\Omega = 0.1$ W/$\text{cm}^2 \cdot \mu\text{m}$, vertical extent of clouds $H_\text{c} = 2$ km, $M_\text{c} = 1$ $\text{mg/m}^3$, and $N_\text{c} = 250$ $\text{cm}^{-3}$.

\begin{figure*}[ht!]
    \centering
    \includegraphics[width=\textwidth]{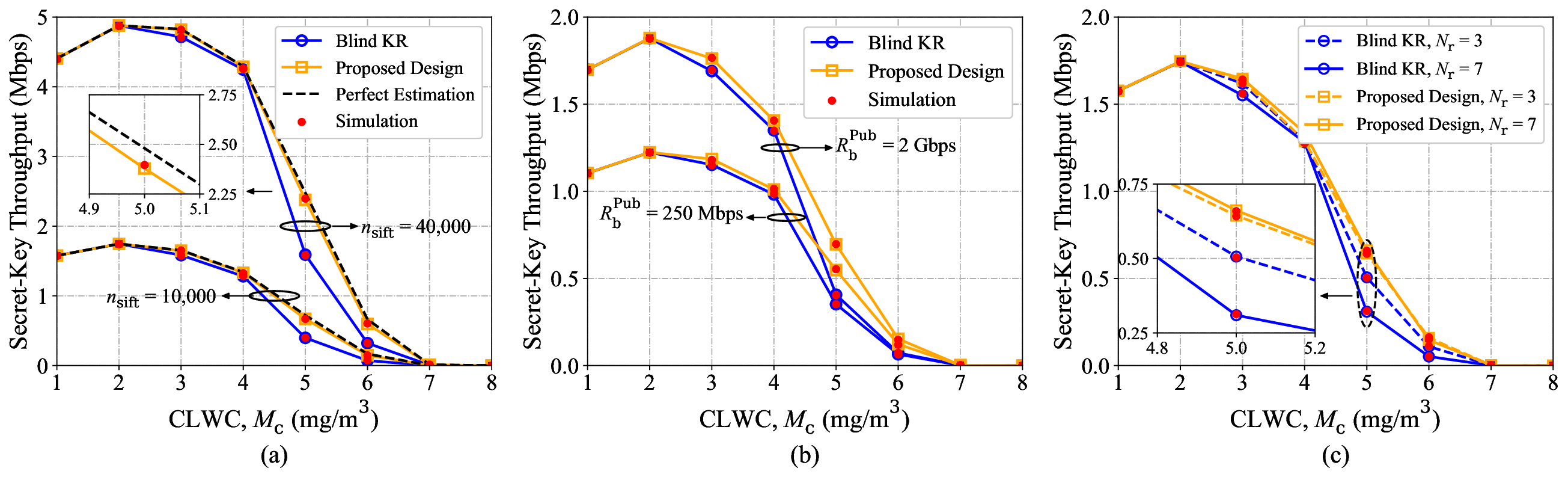}
    \caption{SKT performance of the proposed design and the blind reconciliation versus CLWC with different (a) block lengths, $n_\text{sift}$, (b) data rates of the public channel $R^\text{Pub}_\text{b}$, and (c) maximum number of code rates in the RC-LDPC code family, $N_\text{r}$.}
    \label{fig:perf_comp}
\end{figure*}

We consider two protograph RC-LDPC code families, whose block lengths are $10,000$ and $40,000$, respectively. These families are constructed from the same protograph derived in Appendix~\ref{app:code_design_example}. We also consider the two-step lifting with the progressive edge growth (PEG) algorithm and the circulant PEG algorithm \cite{cuong2024toward}. The first lifting factor is 4; the second lifting factor is 125 and 500, respectively. Note that for the sake of analysis, we consider five code rates of the family, i.e., $\{0.9, 0.8, 0.7, 0.6, 0.5\}$. The thresholds of these code rates are determined based on the empirical performance of these LDPC codes with specific lengths. Notably, the thresholds for the block length of $10,000$ and $40,000$ are $\{0.005, 0.0135, 0.027, 0.054, 0.0835\}$ and $\{0.0065, 0.0175, 0.0335, 0.0635, 0.0896\}$, respectively.

\subsection{Performance Comparison}
\label{sec:perf_comp}

First, we highlight the effectiveness of our proposed design with the conventional blind reconciliation in Fig.~\ref{fig:perf_comp}. Here, the conventional blind reconciliation considers the same RC-LDPC code family as our proposal. Particularly, Fig.~\ref{fig:perf_comp}(a) presents the performance with different sifted block lengths, i.e., $10,000$ and $40,000$. For ease of comparison, we also plot the SKT in the case of perfect error estimation. As expected, our proposed design outperforms blind reconciliation across the considered range of CLWC values. This is because the additional syndrome-based error estimation can help reduce the number of required communication rounds, resulting in improved throughput. Another observation from Fig.~\ref{fig:perf_comp}(a) is that the SKT performance can be improved by increasing the sifted block length. This is due to the required time to sift and reconcile a sifted block being much longer than the time to gather a sifted block over the quantum phase. Therefore, long block lengths can effectively utilize the quantum phase time. Additionally, the Monte Carlo simulations demonstrate a good match with the analytical data, further confirming the accuracy of our theoretical framework. 

Fig.~\ref{fig:perf_comp}(b) illustrates the performance comparison with different values of the public channel data rate, which are $250$ Mbps and $2$ Gbps. As expected, we can enhance the SKT performance by using a high value of $R^\text{Pub}_\text{b}$. The reason is that high values of $R^\text{Pub}_\text{b}$ significantly reduce the transmission delay of the sifted data and syndrome bits, leading to a lower post-processing delay. Furthermore, as the propagation delay is the dominant factor of the post-processing delay, considering the syndrome-based error estimation method can substantially improve the system performance.

Fig.~\ref{fig:perf_comp}(c) investigates the performance of the proposed design and blind reconciliation with different numbers of code rates in the RC-LDPC code family. Specifically, we consider two LDPC code families with three and seven code rates. The three-code family includes $\{0.9, 0.5, 0.7\}$, and the seven-code family additionally considers $\{0.75, 0.65\}$, whose thresholds are $\{0.0205, 0.037\}$. It is observed that when the number of code rates is seven, the blind reconciliation experiences performance degradation due to the increasing number of communication rounds. On the other hand, the proposed design performance with seven code rates is better than that with three code rates. This happens because more code rates can maintain high key reconciliation efficiency over a wide range of QBER, thus improving the SKT performance. This highlights another advantage of our proposed design compared to blind reconciliation.

\subsection{Impact and Selection of Proposed Design Parameters}
\label{sec:para_selection}

\begin{figure}[t]
    \centering
    \includegraphics[width=.8\linewidth]{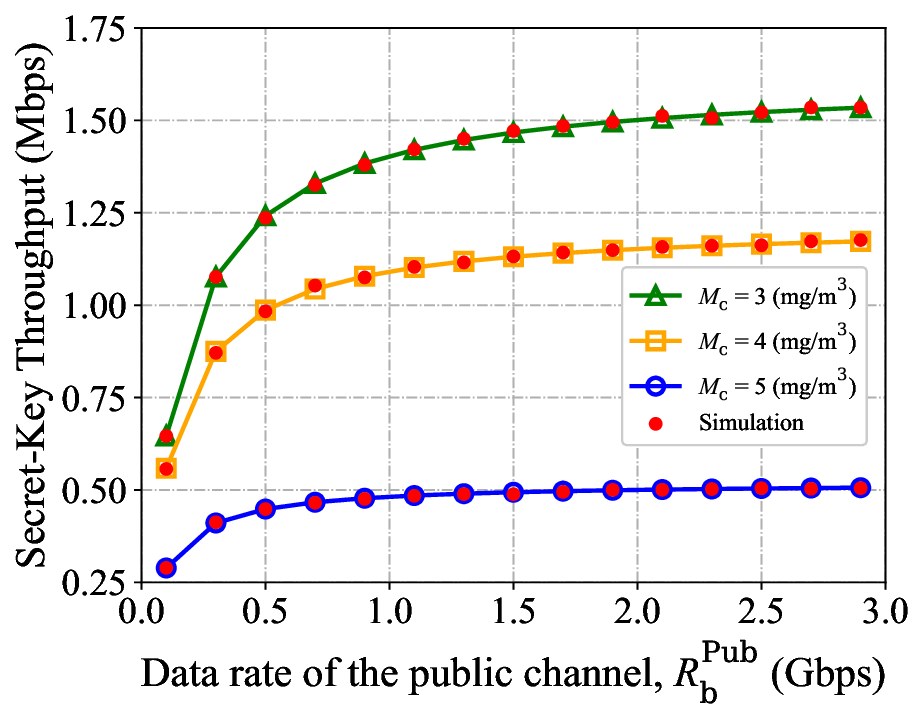}
    \caption{Throughput of the proposed design versus the data rate of the public channel with different CLWC values, $M_\text{c}$.}
    \label{fig:rf_gp_blocklen}
\end{figure}

Our next focus is on the impact and proper selection of the public channel's data rate. Fig.~\ref{fig:rf_gp_blocklen} investigates the SKT of the proposed system with different public channels' data rates. As discussed in Fig.~\ref{fig:perf_comp}(b), increasing the data rate can significantly advance the SKT performance. However, at a certain value of $R^\text{Pub}_\text{b}$, the transmission delay becomes insignificant, and the performance saturates. Based on this finding, we can select the minimum value of $R^\text{Pub}_\text{b}$ to maintain a specific level of throughput. For instance, when the CLWC value is $3$ $\text{mg/m}^3$, we can select the data rate of the public channel as $0.5$ Gbps to retain the throughput levels of $1.25$ Mbps.

\begin{figure}[t]
    \centering
    \includegraphics[width=.8\linewidth]{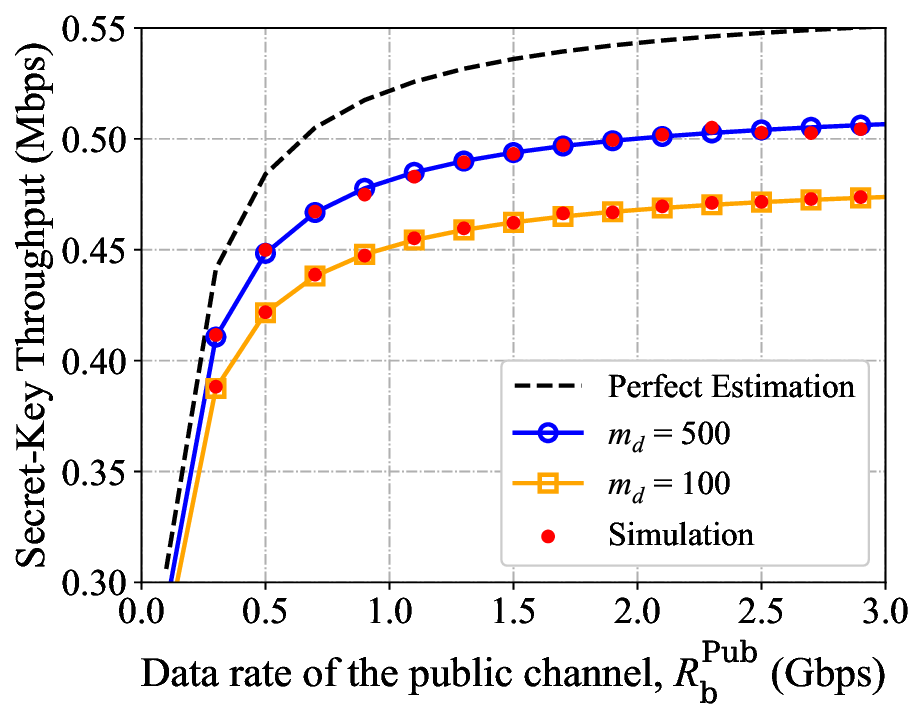}
    \caption{Throughput versus the data rate of the public channel with different numbers of syndrome bits considered for error estimation, $m_d$.}
    \label{fig:diff_m_d}
\end{figure}

We analyze the impact of imperfect error estimation in Fig.~\ref{fig:diff_m_d}, which illustrates the throughput performance with different numbers of syndrome bits used for error estimation, $m_d$. We also set $M_\text{c} = 5$ $\text{mg/m}^3$. As shown in the figure, when the data rate of the public channel is high, i.e., $R^\text{Pub}_\text{b} > 500$ Mbps, the gap between the perfect estimation case and the actual performance becomes more pronounced. The reason is that propagation delays dominate the transmission delays in this region, making the underestimation a major issue. Therefore, imperfect error estimation should be carefully considered when designing practical QKD systems. Moreover, involving a maximum number of possible syndrome bits in error estimation is recommended, as it will lead to less performance degradation thanks to more accurate estimations.

\subsection{Performance Investigation Over A Satellite Pass}
\label{sec:perf_sat_pass}

\begin{figure*}[ht]
    \centering
    \includegraphics[width=\linewidth]{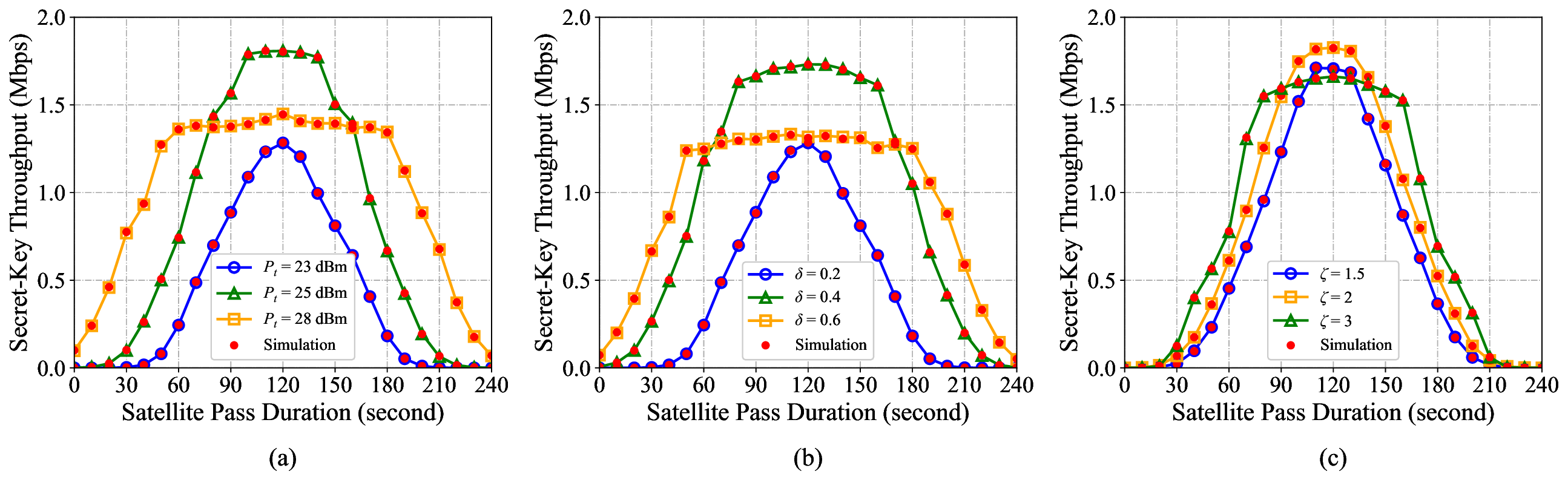}
    \caption{Throughput performance over the satellite pass with different (a) LEO transmitted power, $P_\text{t}$, (b) modulation depths, $\delta$, and (c) DT scale coefficients, $\zeta$.}
    \label{fig:sat_pass_duration}
\end{figure*}

This section investigates the performance of the proposed system in a case study involving Starlink's LEO satellite and a moving vehicle. In particular, we select Alice as the STARLINK-1293 flying over Japan on December 23, 2021. Note that the slant distance and the zenith angle of the system change over time. To investigate the system performance over the satellite pass duration, we re-calculate the slant distance $L$ and the zenith angle $\xi$ for each time instant based on the vehicle’s position and the two-line element set (TLE) of the satellite. The TLE of STARLINK-1293 can be found in \cite{celestrak}, and the method to compute $L$ and $\xi$ can be found in the Appendix of \cite{cuong2024toward}. The information on the vehicle positions needed for the calculation, i.e., longitudes, latitudes, and elevations, can also be found on \cite[Table IV]{cuong2024toward}. The default CLWC value for all the figures is $3$ $\text{mg/m}^3$. Other parameters are the same as the default parameters.

Fig.~\ref{fig:sat_pass_duration}(a) shows the throughput performance of the proposal over the satellite pass with different transmitted power, i.e., $23$, $25$, and $28$ dBm. Note that $t=0$ corresponds to 16:09:20 (UTC+9), December 23, 2021. It can be seen that when the transmitted power is increased from $23$ to $25$ dBm, the system's throughput is enhanced during the satellite pass. It can be explained that high transmitted power can reduce the average QBER at Bob. Therefore, sifted blocks can be reconciled with high-rate LDPC codes. However, if the transmitted power is too high, e.g., $P_\text{t} = 28$ dBm in this case, the signal quality at Eve can also be improved, resulting in more leaked information and a lower throughput. From this result, we can determine the transmitted power that maximizes throughput during the satellite pass. For instance, in this case, we can select the transmitted power as $P_\text{t} = 25$ dBm when $t \in \left[ 80, 160 \right]$, $P_\text{t} = 28$ dBm otherwise.

Fig.~\ref{fig:sat_pass_duration}(b) presents the throughput performance with different modulation depths, which are $0.2$, $0.4$, and $0.6$. The satellite's transmitted power is set to $23$ dBm. As expected, the system performance with $\delta = 0.4$ is higher than that with $\delta = 0.2$. This is because the higher the modulation depth value, the lower the QBER at Bob. However, the system performance with $\delta = 0.6$ is lower compared to the one with $\delta = 0.4$ during $\left[ 70, 170 \right]$. This occurs because the average QBER at Eve also becomes better, thus increasing the amount of information that Eve can obtain via the quantum channel, $I_\text{E}$. Using this figure, Alice can adjust the modulation depth throughout the satellite pass to optimize the system performance. Particularly, when $t \in \left[ 70, 170 \right]$, the modulation depth should be set as $0.4$; otherwise, Alice should adjust it as $\delta = 0.6$.

Finally, we investigate the performance of the proposed system with different values of the DT scale coefficient, i.e., $1.5$, $2$, and $3$, as depicted in Fig.~\ref{fig:sat_pass_duration}(c). It is observed that the throughput performance when $\zeta=2$ is higher than that when $\zeta = 1.5$. This occurs because Bob's average QBER improves, thereby advancing SKT performance. However, setting a high value of $\zeta$ may lead to a longer duration of the quantum phase. This is because Bob takes a longer time to gather a sifted block, causing a performance bottleneck. For example, when $\zeta = 3$, the system performance is degraded during the period $\left[ 100, 140 \right]$ compared to that when $\zeta = 2$. Utilizing this result, Bob can select the proper value of the DT scale coefficient to maximize the system performance. In this case, $\zeta = 2$ when $t \in \left[ 100, 140 \right]$; $\zeta = 3$ otherwise.   

\section{Conclusion}
\label{sec:conclusion}
This paper presented a blind reconciliation protocol for satellite-based FSO/QKD systems using the protograph RC-LDPC code family and syndrome-based error estimation. We analytically studied the performance of the proposed design in terms of SKT, considering major adverse issues over the FSO link and imperfect error estimation. Monte Carlo simulations were conducted to verify the correctness of the analytical framework. From the numerical results, we derive several findings and design guidelines, as follows.
\begin{enumerate}
    \item The proposed design outperforms the conventional blind reconciliation in terms of SKT, especially when the public channel's data rate and the number of code members are high. The trade-off is increased complexity due to the additional error estimation step.
    \item We studied the impact of the public channel's data rate and imperfect error estimation on the system performance. Moreover, we discuss the selection of the public channel's data rate at different CLWC values.
    \item The performance of the proposed system during a satellite pass could be optimized by properly selecting transmitter power, modulation depths, or DT scale coefficients at different time durations.
\end{enumerate}
As future work, the practical feasibility of the proposed approach will be investigated, where the protocol is being validated using experimental data from a LEO-to-ground optical link, with initial efforts reported in \cite{cuong2025b_nict}.

{\appendices
\section{RC-LDPC Code Family Design Example}
\label{app:code_design_example}

This section presents the design of an RC-LDPC code family with rates $\{0.9, 0.85, \cdots, 0.55, 0.5\}$ using the code extension method. In this regard, we start from a protomatrix of the mother code and gradually extend it to obtain the protomatrix of each code member. At each step of extension, the extension part is determined by exhaustively searching for the one with the lowest decoding threshold, subject to constraints imposed by error-floor considerations \cite{van2012design}. The detailed steps of the design procedure can be found in \cite{cuong2024blind}.

Regarding our example, we first construct the protomatrix of the highest code rate, i.e., $0.9$. The size of this matrix is selected as two rows and twenty columns. Considering the size of the base matrix and the linear minimum distance growth property, the R4JA code family \cite{divsalar2009capacity} is a good choice to start with. Then, we apply the lengthening method to the rate-$1/2$ R4JA code to construct the desired protomatrix \cite{van2012design}. After a few steps of lengthening, the base matrix is obtained as 
\begin{align}
    \label{eqn:9_10}
    & \textbf{B}_{9/10} = \nonumber \\
    &
    \resizebox{.95\hsize}{!}{%
    $\begin{bmatrix}
        1 & 1 & 1 & 3 & 1 & 2 & 1 & 3 & 1 & 1 & 3 & 1 & 1 & 1 & 3 & 1 & 1 & 1 & 3 & 1 \\
        1 & 2 & 2 & 1 & 2 & 1 & 3 & 1 & 2 & 2 & 1 & 2 & 2 & 2 & 1 & 2 & 2 & 2 & 1 & 3
    \end{bmatrix}.$
    }
\end{align}
To construct the next code rate, we extend the protomatrix by one row and exhaustively search for the best code. For each row added to the base matrix, the code rate is $\frac{20 - (n_\text{add}+2)}{20}$ where $n_\text{add}$ is the number of added rows. We simplify the search process by constraining the possible values for the first column to $\{0\}$, and for other columns to $\{0, 1\}$. After several steps of extending and searching, we derive the protomatrix of the code rate $1/2$ as
\begin{align}
    \label{eqn:1_2}
    & \textbf{B}_{1/2} = \nonumber \\
    &
    \resizebox{.95\hsize}{!}{%
    $\begin{bmatrix}
        1 & 1 & 1 & 3 & 1 & 2 & 1 & 3 & 1 & 1 & 3 & 1 & 1 & 1 & 3 & 1 & 1 & 1 & 3 & 1 \\
        1 & 2 & 2 & 1 & 2 & 1 & 3 & 1 & 2 & 2 & 1 & 2 & 2 & 2 & 1 & 2 & 2 & 2 & 1 & 3 \\
        0 & 0 & 0 & 1 & 0 & 1 & 1 & 1 & 0 & 0 & 1 & 0 & 0 & 0 & 1 & 0 & 1 & 1 & 1 & 1 \\
        0 & 0 & 0 & 1 & 0 & 1 & 1 & 1 & 0 & 0 & 1 & 0 & 0 & 0 & 1 & 1 & 1 & 1 & 1 & 1 \\
        0 & 0 & 0 & 0 & 0 & 1 & 1 & 1 & 0 & 0 & 1 & 0 & 1 & 1 & 1 & 0 & 1 & 1 & 1 & 1 \\
        0 & 0 & 0 & 0 & 0 & 0 & 0 & 0 & 0 & 0 & 1 & 1 & 0 & 1 & 1 & 0 & 1 & 1 & 1 & 1 \\
        0 & 0 & 0 & 0 & 0 & 0 & 0 & 0 & 0 & 1 & 1 & 0 & 0 & 0 & 1 & 0 & 1 & 1 & 1 & 1 \\
        0 & 0 & 0 & 0 & 0 & 0 & 0 & 0 & 1 & 0 & 1 & 0 & 0 & 0 & 1 & 0 & 1 & 1 & 1 & 1 \\
        0 & 0 & 0 & 0 & 0 & 0 & 0 & 0 & 0 & 0 & 1 & 0 & 1 & 0 & 1 & 0 & 1 & 1 & 1 & 1 \\
        0 & 0 & 0 & 0 & 1 & 0 & 0 & 0 & 0 & 0 & 0 & 0 & 0 & 0 & 0 & 0 & 1 & 1 & 1 & 1
    \end{bmatrix}.$
    }
\end{align}

\section{Proof of Lemma \ref{LEM:P_SIFT_QBER}}
\label{sec:proof_p_sift_qber}
\textit{Regarding the average sift probability,} it is defined as the probability that Bob can detect a bit, either '0' or '1', using the DT detection rule in Sec.~\ref{sec:cv_qkd_dtdd}. The sift probability can be calculated as \cite{trinh2018design}
\begin{align}
    \label{eqn:p_sift}
    P_\text{sift} = \! P_{\text{A}, \text{B}} \! \left( \! 0, 0 \! \right) \! + \! P_{\text{A}, \text{B}} \left( \! 0, 1 \!\right) + \! P_{\text{A}, \text{B}} \left( \! 1, 0 \!\right) + \! P_{\text{A}, \text{B}} \left( \! 1, 1 \!\right),
\end{align}
where $P_{\text{A}, \text{B}} \left( a, b\right)$ denotes the joint probability that Alice transmits bit '$a$' and Bob detects bit '$b$' $\left( a, b \in \left\{ 0, 1 \right\} \right)$. These joint probabilities can be given as \cite{trinh2018design}
\begin{align}
    \label{eqn:P_A_B}
    P_{\text{A}, \text{B}} \left( a, b \right) = P_\text{A} \left( a \right) \int\limits_{0}^{\infty} Q \left[ \frac{\Upsilon \left(a, b \right)}{\sigma_\text{n}} \right] f_{h_\text{B}} \left( h_\text{B} \right) \mathrm{d} h_\text{B},
\end{align}
where $P_\text{A} \left( a \right) = 0.5$ represents the probability that Alice sends bit '$a$', $\Upsilon \left(a, b \right) = \left( 1 - 2 b \right) i_a + \left( 2 b - 1 \right) d_b$, $i_a = \frac{\left( 2 a - 1 \right)}{4} \mathfrak{R} \delta P_\text{t} h_\text{B}$, $d_0$ and $d_1$ are the DTs defined in \ref{sec:cv_qkd_dtdd}. Substituting \eqref{eqn:composite_pdf_Bob} into \eqref{eqn:P_A_B}, changing the variable to $h'_\text{B} = \frac{h_\text{B}}{h_\text{c} A_\text{mod}}$, then using the Gaussian-Laguerre polynomial approximation, we can obtain the numerical expression of $P_\text{sift}$ in \eqref{eqn:closed_form_P_sift}.

\textit{As for Eve's error probability,} it can be defined as the probability that Eve detects a bit that is different from Alice's transmitted bit. As a result, it can be expressed as \cite{trinh2018design}
\begin{align}
    \label{eqn:qber_eve}
    p_\text{e} = P_{\text{A}, \text{E}} \left( 0, 1\right) + P_{\text{A}, \text{E}} \left( 1, 0\right),
\end{align}
where $P_{\text{A}, \text{E}} \left( a, b\right)$ denotes the joint probability that Alice transmits bit '$a$' and Eve detects bit '$b$'. As Eve detects the received signals with the optimal threshold $d_\text{E, th}=0$ to maximize the obtained information, the probability can be computed as \cite{trinh2018design}
\begin{align}
    \label{eqn:P_A_E}
    &P_{\text{A}, \text{E}} \left( \! 0, 1 \! \right) \! = \! P_{\text{A}, \text{E}} \left( \! 1, 0 \!\right) \! = \! \frac{1}{2}  \! \int\limits_{0}^{\infty} \! Q \! \left( \! \frac{\frac{1}{4} \mathfrak{R} \delta P_\text{t} h_\text{E}}{\sigma_\text{n}} \! \right) \! f_{h_\text{E}} \! \left( \! h_\text{E} \! \right)  \!\mathrm{d} h_\text{E}.
\end{align}
Similarly, applying \eqref{eqn:composite_pdf_Eve} to \eqref{eqn:P_A_E}, changing the variable to $h'_\text{E} = \frac{h_\text{E}}{h_\text{c} A_\text{mod}}$, and utilizing the Gaussian-Laguerre polynomial approximation, we can derive the numerical expression as in \eqref{eqn:closed_form_ber_eve}. This completes the proof.

\section{Proof of Lemma \ref{LEM:CDF_P}}
\label{sec:proof_cdf_p}

Let $\gamma$ be an arbitrary value in the range of $\left[ 0, 0.5 \right]$. For a given value of $\rho$, the probability that the estimate, $\hat{\rho}$, is smaller than $\gamma$ is
\begin{align}
    \label{eqn:cdf_p}
    P \left( \hat{\rho} \leq \gamma \middle| \rho \right) = \text{Pr} \! \left[ \! \phi^{-1} \! \left( \! \hat{\rho} \! \right)  \!\leq\! \phi^{-1}\! \left( \gamma \right) \right] \!=\! \text{Pr} \left[ \! w \! \leq \! \phi^{-1} \! \left( \! \gamma\! \right) \right],
\end{align}
The first transformation is conducted since the function $\phi \left( \cdot \right)$ is strictly increasing in the interval $\left[ 0, 0.5 \right]$. Let $W$ be the random variable representing the weight of the syndrome $\textbf{S}$, whose probability mass function (PMF) is written as
\begin{align}
    \label{eqn:pmf_W}
    f_W \left( w; \rho, m \right) = 
    \left(
    \begin{array}{c}
      m \\
      w
    \end{array}
  \right) \phi^w \left( \rho \right) \left[ 1 - \phi^w \left( \rho \right) \right]^{m - w},
\end{align}
for $0 \leq w \leq m$. It is observed that the random variable $W$ follows the binomial distribution. i.e., $W \sim \text{Binomial} \left( m, f \left( \rho \right) \right)$. As a result, the CDF of $W$ is expressed as \cite{wadsworth1960introduction}
\begin{align}
    \label{eqn:cdf_w}
    F_W \left( k; \rho, m \right) = \text{Pr} \left[ W \leq k \right] = I_{1 - \phi \left( \rho \right)} \left( m - k, k + 1 \right),
\end{align}
where $I_\cdot \left( \cdot, \cdot \right)$ is the regularized incomplete beta function. Combining \eqref{eqn:cdf_p} and \eqref{eqn:cdf_w} yields \eqref{eqn:cdf_p_final} and completes the proof.
} 


\bibliographystyle{IEEEtran}
\bibliography{references}

\end{document}